\documentclass[10pt]{article}
\usepackage[nofonts]{dgleich-common}
\usepackage[T1]{fontenc}
\usepackage{todonotes}

\RequirePackage[scaled=0.8]{helvet}%
\RequirePackage[scaled=0.75]{beramono}%
\def\lstbasicfont{\fontfamily{fvm}\selectfont}
\makeatletter  
\setboolean{@dgleich@sanssmallcaps}{true}
\makeatother
\usepackage{xcolor}
\definecolor{firebrick}{HTML}{871a1a}
\definecolor{metro_teal}{HTML}{23373b}
\definecolor{light_teal}{HTML}{7E9AA1}
\usepackage{amsmath}
\usepackage{amsthm}
\usepackage{amssymb}
\usepackage{bm}
\usepackage{graphicx}
\usepackage{authblk}
\graphicspath{{../fig/}{fig/}{assets/}{../assets/}}
\usepackage{substr}
\usepackage[ruled,vlined]{algorithm2e}

\usenatbib
\usehyperref

    \usepackage[capitalise]{cleveref}
    \usepackage{etoolbox}
    
    \usepackage{graphicx}
    
    \newtoggle{preprint}
    \toggletrue{preprint}

    \newtoggle{comments}
    \toggletrue{comments}    

\newcommand{\vd}{\mathbf{d}}

\newcommand{\vtheta}{\boldsymbol{\theta}}
\newcommand{\vz}{\mathbf{z}}
\newcommand{\vp}{\mathbf{p}}

\newcommand{\mA}{\mathbf{A}}

\newcommand{\cL}{\mathcal{L}}
\newcommand{\cT}{\mathcal{T}}
\newcommand{\cR}{\mathcal{R}}

\newcommand{\cP}{\mathcal{P}}
\newcommand{\cA}{\mathcal{A}}

\newcommand{\vv}{\mathbf{v}}
\newcommand{\vy}{\mathbf{y}}
\newcommand{\ve}{\mathbf{e}}

\newcommand{\vs}{\mathbf{s}}

\newcommand{\vbeta}{\boldsymbol{\beta}}
\newcommand{\vgamma}{\boldsymbol{\gamma}}
\newcommand{\vomega}{\boldsymbol{\omega}}
\newcommand{\vol}{\text{\textbf{vol}}}
\newcommand{\cut}{\text{\textbf{cut}}}

\newtheorem*{conj*}{Conjecture}

\newtheorem*{dfn*}{Definition}

\newtheorem*{asm*}{Assumption}
\newtheorem*{thm*}{Theorem}
\newtheorem*{lm*}{Lemma}
\newtheorem*{cor*}{Corollary}
\newtheorem*{clm*}{Claim}
\newtheorem*{fct*}{Fact}

\newtheorem{lemma*}{Lemma}
\newtheorem{prop}{Proposition}
\newtheorem{prop*}{Proposition}

\theoremstyle{remark}

\theoremstyle{remark}

\newcommand\abs[1]{\left|#1\right|}
\newcommand\E[0]{\mathbb{E}}
\newcommand\R[0]{\mathbb{R}}

\newcommand{\argmax}{\operatornamewithlimits{argmax}}

\newcommand\prob[0]{\mathbb{P}}
\newcommand{\norm}[1]{\left\lVert#1\right\rVert}

\newcommand{\bracket}[1]{\langle #1 \rangle}

\definecolor{comment_purple}{HTML}{4700b3}

\begin{document}

\title{Generative Hypergraph Clustering: From Blockmodels to Modularity}
\author{Philip S. Chodrow, Nate Veldt, and Austin R. Benson}

\maketitle

\marginnote[450pt]{\fontsize{7}{9}\selectfont
\noindent Philip S. Chodrow \\ 
Department of Mathematics \\ 
University of California, Los Angeles\\ 
\texttt{phil@math.ucla.edu} \\ 
\noindent Nate Veldt \\ Center for Applied Mathematics \\ Cornell University\\
\texttt{lnv22@cornell.edu } \\ 
\noindent Austin R. Benson \\ 
Department of Computer Science \\ 
Cornell University\\
\texttt{arb@cs.cornell.edu} 
}

\begin{abstract}
	Hypergraphs are a natural modeling paradigm for a wide range of complex relational systems, with nodes representing system components and hyperedges representing multiway interactions. 
	A standard analysis task is to identify clusters of closely related or densely interconnected nodes. 
	In the special case of graphs, in which edges connect exactly two nodes at a time, there are a number of probabilistic generative models and associated inference algorithms for clustering. 
	Many of these are based on variants of the stochastic blockmodel, a random graph with flexible cluster structure. 
	However, there are few models and algorithms for hypergraph clustering. 
	Here, we propose a Poisson degree-corrected hypergraph stochastic blockmodel (DCHSBM), an expressive generative model of clustered hypergraphs with heterogeneous node degrees and edge sizes.  
	Approximate maximum-likelihood inference in the DCHSBM naturally leads to a clustering objective that generalizes the popular modularity objective for graphs. 
	We derive a general Louvain-type algorithm for this objective, as well as a a faster, specialized ``All-Or-Nothing'' (AON) variant in which edges are expected to lie fully within clusters.
	This special case encompasses a recent proposal for modularity in hypergraphs, while also incorporating flexible resolution and edge-size parameters. 
	We show that AON hypergraph Louvain is highly scalable, including as an example an experiment on a synthetic hypergraph of one million nodes. 
	We also demonstrate through synthetic experiments that the detectability regimes for hypergraph community detection differ from methods based on dyadic graph projections. 
	In particular, there are regimes in which hypergraph methods can recover planted partitions even though graph based methods necessarily fail. 
	We use our generative model to analyze different patterns of higher-order structure in school contact networks, U.S. congressional bill cosponsorship, U.S. congressional committees, product categories in co-purchasing behavior, and hotel locations from web browsing sessions, finding interpretable higher-order structure. 
	We then study the behavior of our AON hypergraph Louvain algorithm, finding that it is able to recover ground truth clusters in empirical data sets exhibiting the corresponding higher-order structure. 
	
\end{abstract}

\section{Introduction}

Graphs are a fundamental abstraction for complex relational systems throughout the sciences~\cite{jackson-networks-book,Easley2010,newman-networks-book}. 
A graph represents components of a system by a set of nodes, and interactions or relationships among these components using edges that connect pairs of nodes.  
Much of the structure in complex data, however, involves higher-order interactions and relationships between more than two entities at once~\cite{Milo-2002-motifs,Benson-2016-hoo,Benson-2018-simplicial,lambiotte2019networks,battiston2020networks,torres2020and}.
Hypergraphs are now a burgeoning paradigm for modeling these and many other systems~\cite{li2018submodular,de2020social,sahasrabuddhe2020modelling,veldt2020minimizing}.
A hypergraph still represents the components by a set of nodes, but the edges (often called \emph{hyperedges}) may connect arbitrary numbers of nodes. 
A graph is a special case of a hypergraph, in which each edge connects exactly two nodes.

Graph clustering is a fundamental task in network science that seeks to describe large graphs by dividing their nodes into closely related or interconnected groups (also called clusters or communities)~\cite{porter2009communities,Benson-2016-hoo,fortunato2016community,li2018submodular}. 
Clustering methods for hypergraphs have applications in parallel computation~\cite{Ballard:2016:HPS:3012407.3015144,kabiljo2017social}, circuit design~\cite{karypis1999multilevel}, image segmentation~\cite{Agarwal2005beyond}, semisupervised learning~\cite{Zhou2006learning,yadati2019hypergcn}, 
and higher-order network analysis of gene expression~\cite{tian2009gene}, food webs~\cite{li2017inhomogeneous}, and online social communities~\cite{neubauer2009towards,tsourakakis2017scalable}.

A well-established graph clustering approach is to model the graph as a sample from a probabilistic generative model , in which case the clustering task can be recast as a statistical inference problem~\cite{nowicki2001estimation,hoff2002latent,airoldi2008mixed,karrer2011stochastic,yang2013overlapping,peixoto2014hierarchical,athreya2017statistical}. 
While generative modeling is a mainstay in graph clustering, generative techniques for hypergraphs are largely lacking. 
Indeed, while a small number of generative models of clustered hypergraphs have been proposed~\cite{ghoshdastidar2014consistency,kim2018stochastic,Angelini2016, ke2019community}, these models typically generate hypergraphs with edges of only one size. 
With a recent exception \cite{ke2019community}, these models also do not model degree heterogeneity between nodes.
Heterogeneity in edge size and node degree are both key features of empirical data~\cite{Benson-2018-simplicial}, and their omission limits the applicability of many of these models for practical data analysis. 
An alternative to generative hypergraph modeling is to transform the hypergraph into a dyadic graph via clique expansion, where a dyadic edge connects any pair of nodes that appear together in some hyperedge~\cite{Zhou2006learning,Benson-2016-hoo}. 
While this enables the use of a wide array of existing models and algorithms for graphs, the higher-order structure is lost \cite{chodrow2019configuration}, and generative models of the resulting dyadic graph may rely on explicitly violated independence assumptions. 
Recently, nongenerative approaches based on the popular modularity clustering objective for graphs~\cite{PhysRevE.69.026113} have been proposed for hypergraphs~\cite{kumar2018,Kami2018,veldt2020parameterized}, although their lack of connection to a generative model limits their interpretability. 

Another approach to generative clustering is to use the representation of a hypergraph as a bipartite graph, and apply a generative model (e.g. \cite{larremore2014efficiently,gerlach2018network, yen2020community}) to the latter representation. 
This approach, while appropriate in many data sets, involves a strong assumption: the memberships of any two nodes in a given hyperedge are independent, conditional on the model parameters. 
This assumption is natural for certain classes of data. 
For example, consider an event co-attendance network, with  nodes representing music enthusiasts and hyperedges representing concerts.
Node membership in a hyperedge corresponds to attendance at the specified event.  
To a reasonable approximation, the decision of two fans to attend a given concert may indeed be independent, conditioned on the popularity of the performers, the location of the venue, and so on. 
In other data sets, however, the conditional independence assumption is explicitly violated. 
Multiway social interaction networks give one important class of examples. 
Interactions such as gossip, for instance, normally take place only between trusted individuals. 
The presence of a single uninvited outsider may entirely prevent the interaction from taking place. 
The ``all-or-nothing'' structure of such interactions is an important violation of the conditional independence assumptions made by most bipartite generative models. 
These examples highlight that the task of matching assumptions to higher-order data is an ongoing challenge, for which we benefit from a diversity of distinct tools. 

Here, we propose a generative approach to hypergraph clustering based on a degree-corrected hypergraph stochastic blockmodel (DCHSBM). 
This model generates clustered hypergraphs with heterogeneous degree distributions and hyperedge sizes.
We outline an approximate coordinate-ascent maximum-likelihood estimation scheme for fitting this model to hypergraph data, and show that one stage of this scheme generalizes the well-studied modularity objective for graphs. 
We derive accompanying Louvain algorithms for this class of modularity-like objectives, which are highly scalable in an important special case. 
We show computationally that hypergraph clustering methods are able to detect planted clusters in regimes in which graph-based methods necessarily fail due to known theoretical limits. 
We also show that, in data sets with appropriately matched higher-order structure, our generative hypergraph techniques are able to recover clusters correlated to metadata at higher rates than graph-based techniques. 
Our results highlight the importance of matching generative models to data sets, and point toward a number of directions for further work in higher-order network science.

\section{The Degree-Corrected Hypergraph Stochastic Blockmodel} \label{sec:DCHSBM}

The degree-corrected stochastic blockmodel is a generative model of graphs with both community structure and heterogeneous degree sequences \cite{karrer2011stochastic}. 
We now extend this model to the case of hypergraphs. 

For our model, let $n$ be the number of nodes in a hypergraph.
Each node $i$ is assigned to one of $\bar{\ell}$ groups.
We let $z_i \in [\bar{\ell}] = \{1, 2, \hdots, \bar{\ell}\},$\footnote{Here and elsewhere, we use the notation $[r] = \{1, \ldots, r\}$.} denote the group assignment of node $i$, and collect these assignments in a vector $\vz\in [\bar{\ell}]^n$. 
As in the dyadic degree-corrected SBM, each node $i$ is assigned a parameter $\theta_i$ governing its degree, and
we collect these parameters in a vector $\vtheta \in \R^n$. 
Let $\mathcal{R}$ represent the set of unordered node tuples, so that each $R \in \mathcal{R}$ is a set of nodes representing the location of a possible hyperedge. 
(Following the standard choice for the degree-corrected SBM in graphs, we allow $\mathcal{R}$ to include node tuples with repeated nodes.)
Let $\vz_R$ denote the vector of cluster labels for nodes in a given tuple $R$, and $\vtheta_R$ the vector of degree parameters. 

We use an \emph{affinity function} $\Omega$ to control the probability of placing a hyperedge at a given node tuple $R$,
which depends on the group memberships of the nodes in $R$.
Formally, $\Omega$ maps the group assignments $\vz_R$ to a nonnegative number.
If $\Omega(\vz_R)$ is large, there is a higher probability that a hyperedge forms between the nodes in $R$.
In our model, the number of hyperedges placed at $R \in \mathcal{R}$ is distributed as $a_R \sim \text{Poisson}\left(b_R\pi(\vtheta_R)\Omega(\vz_R)\right)$, where $b_R$ denotes the number of distinct ways to order the nodes of $R$ and $\pi(\vtheta_R) = \prod_{i \in R} \theta_i$ is the product of degree parameters.
The probability of realizing a given value $a_R$ is then
\begin{align}
	\prob(a_R|\vz, \Omega, \vtheta) = \frac{e^{-b_R\pi(\vtheta_R)\Omega(\vz_R)}(b_R\pi(\vtheta_R)\Omega(\vz_R))^{a_R}}{a_R!}.
\end{align} 
This edge generation process has the following intuitive interpretation: for each of the $b_R$ possible orderings of nodes in $R$, we attempt to place a $\text{Poisson}\left(\pi(\vtheta_R)\Omega(\vz_R)\right)$ number of hyperedges on this tuple. 
The result is a weighted hyperedge on the unordered tuple $R$, whose weight can be any nonnegative integer. 
This is a helpful modeling feature, as many empirical hypergraphs contain multiple hyperedges between the same set of nodes.\footnote{Even in hypergraph data sets where we only know the presence or absence of hyperedges (but no weights),
	the Poisson-based model serves as a computationally convenient approximation to a Bernoulli-based model.}
The probability of realizing a given hyperedge set $\mA = (a_R)_{R \in \mathcal{R}}$ is then just the product of probabilities over each $R \in \mathcal{R}$.

\subsection{Estimation of Degree and Affinity Parameters} \label{sec:parameter-estimation}

There are many methods for inference in stochastic blockmodels and their relatives, including variational coordinate ascent \cite{airoldi2008mixed}, variational belief propagation \cite{decelle2011asymptotic,zhang2014scalable}, and Markov Chain Monte Carlo \cite{nowicki2001estimation, peixoto2020merge}. 
We perform approximate maximum-likelihood inference via coordinate ascent.
We do so in order to exploit a recent connection between maximum-likelihood inference in stochastic blockmodels and the popular modularity objective for graph clustering \cite{newman2016equivalence}.
Our coordinate ascent framework, in which we alternate between estimating parameters and node labels, is a close relative of expectation-maximization (EM) algorithms for blockmodel inference \cite{decelle2011asymptotic}. 
Standard versions of EM construct ``soft'' clusters, in which each node is given a weighted assignment in every possible cluster. 
``The'' cluster for a given node is often taken to be the cluster in which the node has largest weight. 
In contrast, our approach generates ``hard'' clusters in which each node belongs to exactly one cluster. 
Profile likelihood methods offer an alternative framework for maximum-likelihood inference in SBMs \cite{bickel2009nonparametric}, and their development for hypergraphs is another promising avenue of future work.   

In the maximum-likelihood framework, we learn estimates $\hat{\vz}$ of the node labels,  $\hat{\Omega}$ of the affinity function, and $\hat{\vtheta}$ of the degree parameters by solving the optimization problem 
\begin{align}
\hat{\vz}, \; \hat{\Omega}, \; \hat{\vtheta} \equiv  \argmax_{\vz, \Omega, \vtheta}\;\;  \prob(\mA|\vz, \Omega, \vtheta)\;, \label{eq:ml}
\end{align}
where $\mA$ is a given data set represented by a collection of (integer-weighted) hyperedges.
As usual, it is easier to work with the log-likelihood, which has the same local optima.
The log-likelihood is
\begin{align}
\cL(\vz, \Omega, \vtheta) = \sum_{R\in \cR} \log \prob(a_R|\vz, \Omega, \vtheta) 
= Q(\vz, \Omega, \vtheta) + K(\vtheta) + C\;, \label{eq:hypergraph_ll}
\end{align}
where 
\begin{align}
Q(\vz, \Omega, \vtheta) &\equiv \sum_{R\in \cR}\left[a_R\log \Omega(\vz_R) -b_R\pi(\vtheta_R)\Omega(\vz_R)  \right] \\ 
K(\vtheta) &\equiv \sum_{R\in \cR}a_R \log \pi(\vtheta_R)\\ 
C &\equiv \sum_{R\in\cR}\left[a_R \log b_R - \log a_R!\right]\;.
\end{align}
The first term $Q(\vz, \Omega, \vtheta)$ is the only part of the log-likelihood that depends on the group assignments $\vz$ and affinity function $\Omega$. 
The second term depends on $\vtheta$, while the third term depends only on the data $\mA$ and can be disregarded for inferential purposes. 


In the coordinate ascent approach to maximum-likelihood, we alternate between two stages. 
In the first stage, we assume a current estimate $\hat{\vz}$ and obtain new estimates of $\Omega$ and $\vtheta$ by solving 
\begin{align}
	\hat{\Omega}, \hat{\vtheta} = \argmax_{\Omega, \vtheta} \cL(\hat{\vz}, \Omega, \vtheta)\;. \label{eq:coordinate-ascent-1}
\end{align}
The resulting $\hat{\Omega}, \hat{\vtheta}$ can be viewed as maximum-likelihood estimates, \emph{conditioned} on the current estimate $\hat{\vz}$ of the label vector $\vz$. 
In the second stage, we assume current estimates $\hat{\Omega}$ and $\hat{\vtheta}$ and obtain a new estimate of $\vz$ by solving 
\begin{align}
	\hat{\vz} = \argmax_{\vz} \cL(\vz, \hat{\Omega}, \hat{\vtheta})\;. \label{eq:coordinate-ascent-2}
\end{align}
We alternate between these two stages until convergence.

There are several identifiability issues that must be addressed. 
First, permuting the group labels in $\vz$ and $\Omega$ does not alter the value of the likelihood. 
We therefore impose an arbitrary order on group labels.
Second, the number of possible groups $\bar{\ell}$ can in principle be larger than the number of groups present in $\vz$. 
Such a case would correspond to the presence of groups which are statistically possible but empty in the given data realization. 
While other treatments are possible, we choose to disregard empty groups and treat $\bar{\ell}$ as equal to the number of distinct labels in an estimate of $\vz$. 
A final form of unidentifiability relates to the scales of $\vtheta$ and $\Omega$. 
For a fixed $\vtheta$ and $\Omega$, we can construct $\vtheta' \neq \vtheta$ and $\Omega' \neq \Omega$ such that $\cL(\vz, \Omega, \vtheta) = \cL(\vz, \Omega', \vtheta')$ (\Cref{sec:unidentifiable}). 
To enforce identifiability, we must therefore place a joint normalization condition on either $\vtheta$ or $\Omega$. 
We choose to constrain $\vtheta$ such that 
\begin{align}
\sum_{i = 1}^n\theta_i\delta(z_i,\ell) = \vol(\ell),\;\; \ell = 1, \ldots, \bar{\ell}\;, \label{eq:normalization}
\end{align}
where $\vol(\ell) = \sum_{i = 1}^n d_i\delta(z_i, \ell)$
and $d_i$ is the (weighted) number of hyperedges in which node $i$ appears.\footnote{Here, and throughout the rest of the text, $\delta$ is an indicator function evaluating to 1 if all its inputs are equal, and is 0 otherwise.} 

The usefulness of~\eqref{eq:normalization} is that, when $\vz$ is known or estimated, the conditional maximum-likelihood estimates $\hat{\vtheta}$ and $\hat{\Omega}$ in \eqref{eq:coordinate-ascent-1} take simple, closed forms.
First, 
for a fixed label vector $\vz$, when using the normalization~\eqref{eq:normalization}, the maximum-likelihood estimate for $\vtheta$ is 
\begin{align}
\hat{\vtheta} = \vd\;. \label{eq:hat-d}	
\end{align}
(See \Cref{sec:proof_constant}.)
Second, conditioned on $\vz$, if $\Omega$ takes constant value $\omega$ on some set $Y$ of unordered tuples of labels,\footnote{In full generality, we can have one such $\omega$ for every possible label arrangement for each hyperedge size
		in the data. Later, we will make natural restrictions on $\ fOmega$.} the maximum likelihood estimate for $\omega$ is 
	\begin{align}
	\hat{\omega} = \frac{\sum_{\vy\in Y}\sum_{R\in \cR}a_R \delta(\vz_R,\vy)}{\sum_{\vy\in Y} \prod_{y \in \vy} \vol(y)}\;. \label{eq:hat_omega_generalized}
	\end{align}
(See \Cref{sec:omega_estimate}.)
Although \eqref{eq:hat-d} assumes that $\vz$ was fixed, it is not necessary to \emph{know} $\vz$ in order to form the estimate $\hat{\vtheta}$. 
However, forming the estimate $\hat{\omega}$ via \eqref{eq:hat_omega_generalized} requires that we know or estimate $\vz$. 
It is therefore important to remember that $\hat{\omega}$ is not a globally optimal estimate, but rather a locally optimal estimate conditioned on the currently estimated group labels. 

The formula \eqref{eq:hat_omega_generalized} could also be inserted directly into the full likelihood maximization problem \eqref{eq:ml}, eliminating the parameters corresponding to $\Omega$ and producing a lower-dimensional \emph{profile likelihood}, which could then in principle be optimized directly. 
This approach has been successful for dyadic blockmodels \cite{bickel2009nonparametric}, and the development of similar methods for hypergraph blockmodels would be of considerable interest.  
The advantage of our coordinate ascent framework is that we are able to develop fast heuristics for solving~\eqref{eq:coordinate-ascent-2}, by generalizing widely-used algorithms for graph clustering (Section~\ref{sec:hmll}). 
Solving problem~\eqref{eq:coordinate-ascent-1} in our framework can also be interpreted as evaluating the profile likelihood for a fixed cluster vector $\vz$, highlighting the relationship between these approaches.

We now turn to the problem of inferring the label vector $\vz$. 
This problem leads naturally to a class of modularity-type objectives for hypergraph clustering.

\section{Hypergraph Modularities}\label{sec:hyper_mod}

Our results from the previous section imply that the estimated degree parameter $\hat{\vtheta}$ and piecewise constant affinity function $\hat{\Omega}$ can be efficiently estimated in closed form, provided an estimate of $\vz$. 
This provides a solution to the first stage \eqref{eq:coordinate-ascent-1} of coordinate ascent. 
We now discuss the second stage \eqref{eq:coordinate-ascent-2}. 
From \eqref{eq:hypergraph_ll}, it suffices to optimize $Q$ with respect to $\vz$.  
To do so, it is helpful to impose some additional structure on $\hat{\Omega}$. 

\subsection{Symmetric Modularities} \label{sec:symmetric_omega}
We obtain an important class of objective functions by stipulating that $\Omega$ is symmetric with respect to permutations of node labels. 
In this case, $\Omega(\vz_R)$ depends not on the specific labels $\vz_R$ in a given node tuple $R$, but only on the number of repetitions of each. 
Statistically, the corresponding DCHSBM generates hypergraphs in which all groups are statistically identical, conditioned on the degrees of their constituent nodes. 
Symmetric affinity functions thus give a flexible generalization of the planted partition stochastic blockmodel~\cite{jerrum1998metropolis,condon2001algorithms} to the setting of hypergraphs. 

Define the function $\phi(\vz) = \vp$, where $p_j$ is the number of entries of $\vz$ in the $j$th largest group in $\vz$, with ties broken arbitrarily. 
For example, if $\vz = (1, 1, 4, 1, 2, 3, 2)$, then $\vp = (3, 2, 1, 1)$. 
We call $\vp$ a \emph{partition vector}. 
The symmetry assumption implies that $\Omega$ is a function of $\vz_R$ only through $\vp = \phi(\vz_R)$. 
Accordingly, we abuse notation by writing $\Omega(\vp) \equiv \Omega(\vz)$ when $\vp = \phi(\vz)$. 

We now define generalized cuts and volumes corresponding to a possible partition vector $\vp$ for tuples of $k$ nodes:
\begin{align}
\cut_\vp(\vz) &\equiv \sum_{R \in \cR^k}a_R\delta(\vp, \phi(\vz_R)),  \\ 
\vol_\vp(\vz) &\equiv \sum_{\vy \in [\overline{\ell}]^{k}}\delta(\vp, \phi(\vy))\prod_{y\in \vy}\vol(y),
\end{align}
where $\cR^k$ is the subset of tuples in $\cR$ consisting of $k$ nodes.
The function $\cut_\vp(\vz)$ counts the number of edges that are split by $\vz$ into the specified partition $\vp$, while the function $\vol_\vp(\vz)$ is a sum-product of volumes over all  grouping vectors $\vy$ that induce partition $\vp$. 
Let $\mathcal{P}$ be the set of partition vectors on sets up to size $\bar{k}$, the maximum size of hyperedges. 
We show in \cref{sec:mod_cut_vol_deriv} that the symmetric modularity objective can then be written as 
\begin{align}
Q(\vz, \Omega, \vd) &= \sum_{\vp \in \cP} \left[\cut_\vp(\vz)\log\Omega(\vp) - \vol_\vp(\vz)\Omega(\vp)\right] \;. 
\label{eq:combinatorial}
\end{align}
For a partition vector $\vp$ for tuples of $k$ nodes, direct calculation of $\vol_\vp(\vz)$ is a summation of $\overline{\ell}^k$ elements, which can be impractical when either $\overline{\ell}$ or $k$ are large. 
We show in \Cref{sec:eval_sums}, however, that it is possible to evaluate these sums efficiently via a combinatorial identity. We also give a formula for updating volume terms $\mathrm{\vol}_\vp(\vz)$ when a candidate labeling is modified.

The objective function \eqref{eq:combinatorial} is related to the multiway hypergraph cut problem studied by \citet{veldt2020hypergraph}. 
They formulate the hypergraph cut objective in terms of \emph{splitting functions}, which associate penalties when edges are split between two or more clusters. 
One then aims to minimize the sum of penalties subject to constraints that certain nodes must not lie in the same cluster. 
Symmetric affinity functions in our framework correspond to \emph{signature-based splitting functions} in their terminology. 
Table~\ref{tab:symmetric_affinities} lists four of many families of affinity functions. 

\begin{table}   
	\centering
	\begin{tabular}{l l l }
		\toprule
		All-Or-Nothing (AON)            & $\Omega(\vp) = 
		\begin{cases}
		\omega_{k1} &\quad \norm{\vp}_{0} = 1 \\ 
		\omega_{k0} &\quad \text{otherwise.}
		\end{cases}$ & \\
		Group Number (GN)      & $\Omega(\vp) = f\left(\norm{\vp}_0, k\right)$ & \\
		Relative Plurality (RP) & $\Omega(\vp)= g\left(p_1 - p_2, k\right)$  & \\
		Pairwise  (P)& $\Omega(\vp)= h\left(\sum_{i \neq j}p_{i}p_{j}, k\right)$  & \\
		\bottomrule
	\end{tabular}
	\bigskip
	\caption{
		Symmetric affinity functions. 
		Throughout, $k = \norm{\vp}_0$ is the number of nodes in partition $\vp$, $\omega_{k0}$ and
		$\omega_{k1}$ are scalars, and $f$, $g$, and $h$ are arbitrary scalar functions.  
	} 
	\label{tab:symmetric_affinities}          
\end{table}

The All-Or-Nothing affinity function distinguishes only whether or not a given edge is contained entirely within a single cluster. 
This affinity function is especially important for scalable computation, and we discuss it further below. 
The Group Number affinity depends only on the number of distinct groups represented in an edge, regardless of the number of incident nodes in each one. 
Special cases of the Group Number affinity arise frequently in applications. 
When $f\left(\norm{\vp}_0, k\right) = \text{exp}(\norm{\vp}_0 -1)$, the first term of the modularity objective corresponds to a hyperedge cut penalty that is known in the scientific computing literature as ``$\text{connectivity}-1$''~\cite{deveci2015hypergraph}, the $K-1$ metric~\cite{karypis2000multilevel}, or the boundary cut~\cite{hendrickson2000graph}. It has also been called \emph{fanout} in the database literature~\cite{kabiljo2017social}. The related Sum of External Degrees penalty~\cite{karypis2000multilevel} is also a special case of the Group Number affinity. 
The Relative Plurality affinity considers only the relative difference between the size of the largest group represented in an edge and the next largest group. 
This rather specialized affinity function is especially appropriate in contexts where groups are expected to be roughly balanced, as we find, for example, in party affiliations in Congressional committees. 
Finally, the Pairwise affinity counts the number of pairs of nodes within the edge whose clusters differ. 
While this affinity function uses similar information to that used in dyadic graph models, there is no immediately apparent equivalence between any dyadic random graph and a DCHSBM with the Pairwise affinity function. 
There are many more symmetric affinity functions; see Table 3 of~\citet{veldt2020hypergraph} for several other splitting functions which can be used to define affinities. 

An important subtlety was recently raised by \citet{zhang2020statistical} concerning the relationship between blockmodels and modularity in \citet{newman2016equivalence}, which also applies to our derivation of \eqref{eq:combinatorial} above and \eqref{eq:aon-simplified} below. 
We derived the conditional maximum-likelihood estimates \eqref{eq:hat-d} and \eqref{eq:hat_omega_generalized} of $\vtheta$  and $\Omega$  under the assumption of a general, unconstrained affinity function $\Omega$. 
It is not guaranteed that these same estimates maximize the likelihood when additional constraints---such as the symmetry constraint $\Omega(\vz) = \Omega(\phi(\vz))$---are imposed. 
Indeed, as \citet{zhang2020statistical} show for the case of dyadic graphs, \eqref{eq:hat-d} and \eqref{eq:hat_omega_generalized} for estimating $\hat{\vtheta}$ and $\hat{\Omega}$ are only exact under the symmetry assumption on $\Omega$ when $\vol(\ell)$ is constant for each $\ell \in [\hat{\ell}]$. 
When the sizes of groups vary, as is typical in most data sets, \eqref{eq:hat-d} and \eqref{eq:hat_omega_generalized} are instead \emph{approximations} of the exact conditional maximum-likelihood estimates. 
The situation is reminiscent of the tendency of the graph modularity objective to generate clusters of approximately equal sizes \cite{gleich2016mining}. 
The objectives and algorithms we develop below should therefore be understood as approximations to coordinate-ascent maximum likelihood inference, which are exact only in the case that all clusters have equal volumes.   
See \citet{zhang2020statistical} for a detailed discussion of these issues in the context of dyadic graphs.

\subsection{All-Or-Nothing Modularity}
All-Or-Nothing affinity function is of special interest for modeling and computation.  
This affinity function is a natural choice for systems in which the occurrence of an interaction or relationship depends strongly on group homogeneity. 

Inserting the All-Or-Nothing affinity function from Table~\ref{tab:symmetric_affinities} into \eqref{eq:combinatorial} yields, after some algebra (\Cref{sec:mod_cut_vol_deriv_aon}), the objective 
\begin{align}
Q(\vz, \Omega, \vd) = - \sum_{k = 1}^{\overline{k}} \beta_k\left[\cut_{k}(\vz) + \gamma_k \sum_{\ell = 1}^{\overline{\ell}}\vol(\ell)^k \right] + J(\vomega)\;, \label{eq:aon-simplified}
\end{align}
where $\beta_k = \log \omega_{k1} - \log \omega_{k0}$, $\gamma_k = \beta_k^{-1}(\omega_{k1} - \omega_{k0})$,  and $J(\vomega)$ collects terms that do not depend on the partition $\vz$. 
We collect $\{\beta_k\}$ and $\{\gamma_k\}$ into vectors $\vbeta, \vgamma \in \R^{\bar{k}}$. 
We have also defined 
\begin{align}
	\cut_k(\vz) \equiv m_k - \sum_{R \in \cR^k}a_R\delta(\vz_R)\;.
\end{align}
In this expression, $m_k$ is the (weighted) number of hyperedges of size $k$, i.e., $m_k = \sum_{R \in \cR^k}a_R$. 
The cut terms $\cut_k(\vz)$ thus count the number of hyperedges of size $k$ which contain nodes in two or more distinct clusters. 
This calculation is a direct generalization of that from~\citet{newman2016equivalence} for graph modularity. 
Indeed, we recover the standard dyadic modularity objective by restricting to $k = 2$. 
We call \eqref{eq:aon-simplified} the \emph{All-Or-Nothing (AON) hypergraph modularity}. 

Recently, \citet{Kami2018} proposed a ``strict modularity'' objective for hypergraphs. 
This strict modularity is a special case of \eqref{eq:aon-simplified}, obtained by choosing $\omega_{k0}$ and $\omega_{k1}$ such that $\beta_k = 1$ and $\gamma_k = \frac{m_k}{(\vol(H))^k}$, where $\vol(H) = \sum_{i=1}^n d_i$ is the sum of all node degrees in hypergraph $H$. 
However, leaving these parameters free lends important flexibility to our proposed AON objective \eqref{eq:aon-simplified}. 
Tuning $\vbeta$ allows one to specify which hyperedge sizes are considered to be most relevant for clustering. 
In email communications, for example, a very large list of recipients may carry minimal information about their social relationships, 
and it may be desirable to down-weight large hyperedges by tuning $\vbeta$. 
Tuning $\vgamma$ has the effect of modifying the sizes of clusters favored by the objective, in a direct generalization of the resolution parameter in dyadic modularity~\cite{reichardt2006statistical,veldt2018cc}. 
Importantly, it is not necessary to specify the values of these parameters \emph{a priori}; instead, they can be adaptively estimated via \eqref{eq:hat_omega_generalized}. 

\section{Hypergraph Maximum-Likelihood Louvain 
}
\label{sec:hmll}

In order to optimize the modularity objectives \eqref{eq:combinatorial} and \eqref{eq:aon-simplified}, we propose a family of agglomerative clustering algorithms. 
These algorithms greedily improve the specified objective through local updates to the node label vector $\vz$.
The structure of these algorithms is based on the widely used and highly performant Louvain heuristic for graphs~\cite{blondel2008fast}. 
The standard heuristic alternates between two phases. 
In the first phase, each node begins in its own singleton cluster. 
Then, each node $i$ is visited and moved to the cluster of the adjacent node $j$ which maximizes the increase in the objective $Q$. 
If no such move increases the objective, then $i$'s label is not changed. 
This process repeats until no such moves exist which increase the objective. 
In the second phase, a ``supernode'' is formed for each label. 
The supernode represents the set of all nodes sharing that label.   
Then, the first phase is repeated, generating an updated labeling of supernodes, which are then aggregated in the second phase. 
The process repeats until no more improvement is possible. 
Since every step in the first phase improves the objective, the algorithm terminates with a locally optimal cluster vector $\vz$. 

This heuristic generalizes naturally to the setting of hypergraphs. 
However, the incorporation of heterogeneous hyperedge sizes and general affinity functions considerably complicates implementation. 
Here we provide a highly general Hypergraph Maximum-Likelihood Louvain (HMLL) algorithm for optimizing the symmetric modularity objective~\eqref{eq:combinatorial}. 
For the important case of the All-Or-Nothing (AON) affinity, the simplified objective~\eqref{eq:aon-simplified} admits a much simpler and faster specialized Louvain algorithm, which we describe in \Cref{sec:AON_louvain}. 
As we show in subsequent experiments, this specialized algorithm is highly scalable and effective in recovering ground truth clusters in data sets with polyadic structure plausibly modeled by the AON affinity. 

\subsection{Symmetric Hypergraph Maximum-Likelihood Louvain} \label{sec:sym_HMLL}
The first phase of our Symetric HMLL algorithm mirrors standard graph Louvain: nodes begin in singleton clusters and in turn greedily move to adjacent clusters until no more improvement is possible. Phase two of graph Louvain reduces edges between clusters into weighted edges between supernodes in a structure-preserving way. However, in the hypergraph case, naively collapsing clusters and hyperedges would discard important information about hyperedge sizes and the way each hyperedge is partitioned across clusters. Therefore, in subsequent stages of our algorithm, we greedily improve the objective by moving entire sets of nodes in the original hypergraph, rather than greedily moving individual nodes. In this way, a set of nodes that was assigned to the same cluster in a previous iteration is essentially treated as a supernode and moved as a unit, without collapsing the hypergraph and losing needed information about hyperedge structure.

Our overall procedure is formalized in Algorithms~\ref{alg:symHMLLstep} and~\ref{alg:symHMLL}. \Cref{alg:symHMLLstep} visits in turn each set of nodes $S_c$ that represents a cluster $c$ from a previous iteration. The algorithm evaluates the change $\Delta Q$ in the objective function $Q$ associated with moving all of $S_c$ to an adjacent cluster, and then carries out the move that gives the largest positive change to the objective. This is repeated until moving a set $S_c$ can no longer improve the objective. The entire Symetric HMLL method is summarized by running \Cref{alg:symHMLL}, which starts by placing every node in a singleton cluster before calling \Cref{alg:symHMLLstep}  for the first time.

\begin{algorithm}[t]
	\SetAlgoLined
	\DontPrintSemicolon
	\KwData{Hypergraph $H$, affinity function $\Omega$, current label vector $\vz$}
	\KwResult{Updated label vector $\vz'$}\; 
	
	$\mathcal{C} \gets \text{unique}(\vz)$   \tcp*{unique cluster labels in the initial clustering $\vz$}
	$\vz' \gets \vz$  \tcp*{By greedily moving clusters in $\vz$,  form new clustering $\vz'$}
	$\textit{improving} \gets \mathrm{true}$\\
	\While{\text{improving}}{
		$\textit{improving}  \gets \mathrm{false}$\\
		\For{$c$ in $\mathcal{C}$}
		{
			$S_c \gets \{i \colon z_i = c\}$ \tcp*{nodes with label $c$ in original clustering $\vz$} 
			$\mathcal{C}' \gets \text{unique}(\vz')$   \tcp*{cluster labels in the current clustering $\vz'$}
			\tcp{Set of clusters $\mathcal{A}_c$ in $\vz'$ that are adjacent to $S_c$}
			$\mathcal{A}_c \gets \bigcup_{e \in E} \{\tilde{c} \in \mathcal{C}' \colon \exists v \in V \text{ with } z_v' = \tilde{c} \text{ and } i \in S_c \text{ with } v,i \in e \}$ \\
			\tcp{maximum $\Delta$ and maximizer $c'$ of the change in $Q$ from moving set $S_c$ to an adjacent cluster $\tilde{c}$ in $\mathcal{A}_c $}
			$(\Delta, c') \gets \argmax_{\tilde{c} \in \mathcal{A}_c}\Delta Q(H,\Omega, \vz', S_c, \tilde{c})$\;
			\tcp{update $\vz'$ if improvement found}
			\If{$\Delta>0$}{
				\For{$i \in S_c$}{
					$z_i' \gets c'$\\
				}
				$improving \gets \text{true}$
			}
		}
	}
	\Return{$\vz'$}	
	
	\caption{SymmetricHMLLstep($H$, $\Omega$, $\vz$)} \label{alg:symHMLLstep}
\end{algorithm}

\begin{algorithm}[t]
	\SetAlgoLined
	\DontPrintSemicolon
	\SetKwRepeat{Do}{do}{while}
	\KwData{Hypergraph $H$, affinity function $\Omega$}
	\KwResult{Label vector $\vz$}\; 
	$\vz' \gets \vz \gets [n]$ \tcp*{assign each cluster to singleton} 
	\Do{$\vz \neq \vz'$}{
		$\vz \gets \vz'$\\
		$\vz' \gets\text{SymmetricHMLLstep}(H,\Omega,\vz)$}
	\Return{$\vz$}
	\caption{SymmetricHMLL($H$, $\Omega$)} \label{alg:symHMLL}
\end{algorithm}

\Cref{alg:symHMLLstep} relies on a function $\Delta Q$ which computes the change in the objective $Q$ associated with moving all nodes from $S_c$ to an adjacent cluster.  
Changes to the second (volume) term in the objective can be computed efficiently using combinatorial identities (\Cref{sec:eval_sums}, \Cref{prop:recursion}).  
Changes to the first (cut) term require summing across all hyperedges incident to a node or set of nodes. 
At each hyperedge, we must evaluate the affinity $\Omega(\vp)$ on the current partition $\vp$, as well as the affinity $\Omega(\vp')$ associated with the candidate updated partition $\vp'$.  
This situation contrasts with the case of the graph Louvain algorithm, in which it is sufficient to check whether a given edge joins nodes in the same or different clusters. 
The fact that we need to store and update the partition vector $\vp$ for each hyperedge is what prevents us from collapsing a cluster of nodes into a monolithic supernode and recursively applying \Cref{alg:symHMLL} on a reduced data structure, as customary in graph Louvain. 
Thus, while clusters of nodes move as a unit in \Cref{alg:symHMLLstep} as well, it is necessary in this case to operate on the full adjacency data $\cA$ at all stages of the algorithm. 
This can make \Cref{alg:symHMLL} slow on hypergraphs of even modest size. 
Developing efficient algorithms for optimizing the general symmetric modularity objective or various special cases is an important avenue of future work. 

\subsection{All-Or-Nothing Hypergraph Maximum-Likelihood Louvain}

When $\Omega$ is the All-Or-Nothing affinity function, considerable simplification is possible. 
For each edge we need not compute the full partition vector $\vp$ but only check whether or not $\norm{\vp}_0 = 1$. 
Rather than a general affinity function $\Omega$, we instead supply the parameter vectors $\vbeta$ and $\vgamma$ appearing in~\eqref{eq:aon-simplified}. 
This allows us to compute on considerably simplified data structures.
In particular, we are able to follow the classical Louvain strategy of collapsing clusters into single, consolidated ``supernodes,'' and restrict attention to hyperedges that span multiple supernodes. 
Because we do not need to track the precise \emph{way} in which the hyperedges span multiple supernodes, we can forget much of the original adjacency data $\cA$ and instead simply store the edge sizes of the hypergraph. 
These simplifications enable both significant memory savings and very rapid evaluation of the objective update  function $\Delta Q$. 
We provide pseudocode for exploiting these simplifications in \Cref{sec:AON_louvain}. 

\subsection{Number of Clusters}

	Like most Louvain-style modularity methods, the user cannot directly control the number of clusters returned by HMLL. 
	One approach to influence the number of clusters is to manually set values for the affinity function $\Omega$ or the parameters $\vbeta$ and $\vgamma$ (if using the All-Or-Nothing affinity). 
	Rather than inferring these parameters from data, one can set them \emph{a priori} and perform a single round of optimization over $\vz$. 
	This approach generalizes the common treatment of the resolution parameter in dyadic modularity maximization as a tuning hyperparameter rather than a number to be estimated from data \cite{reichardt2006statistical}. 
	Considerable experimentation may be required in order to obtain the desired number of clusters, and retrieving an exact number may not be possible. 

	Another approach to influencing the number of clusters is to impose a Bayesian prior on the community labels. 
	In the simplest version of a Bayesian approach, one assumes that each node is independently assigned one of $\bar{\ell}$ labels with equal probability, prior to sampling edges. 
	The probability of realizing a given label vector $\vz$ is then $\bar{\ell}^{-n}$, which generates a term of the form $-n\log \bar{\ell}$ in the log-likelihood $\cL$. 
	This term may then be incorporated into Louvain implementations, with the result that greedy moves which reduce the total number of clusters $\bar{\ell}$ are strongly incentivized. 
	The resulting regularized algorithm may then label vectors $\vz$ with slightly smaller numbers of distinct clusters. 
	This can be useful when it is known \emph{a priori} that the true number of clusters in the data is small. 
	Our implementation of AON HMLL incorporates this optional regularization term. 
	We use this term only in the synthetic detectability experiments presented in \Cref{fig:detectability}. 

\section{Experiments with Synthetic Data}

\subsection{Runtime}

Dyadic Louvain algorithms are known for being highly efficient in large graphs. 
Here, we show that AON HMLL can achieve similar performance on synthetic data to Graph MLL (GMLL), a variant of the standard dyadic Louvain algorithm in which we return the combination of resolution parameter and partition that yield the highest dyadic likelihood. 
For a fixed number of nodes $n$, we consider a DCHSBM-like hypergraph model with $\bar{\ell} = n/200$ clusters and $m = 10n$ hyperedges with size $k$ uniformly distributed between 2 and 4. 
Each $k$-edge is, with probability $p_k$,  placed uniformly at random on any $k$ nodes within the same cluster. 
Otherwise, with probability $1 - p_k$, the edge is instead placed uniformly at random on \emph{any} set of $k$ nodes. 
We use this model rather than a direct DCHSBM to avoid the computational burden of sampling edges at each $k$-tuple of nodes, which is prohibitive when $n$ is large. 
For the purpose of performance testing, we compute estimates of the parameter vectors $\vbeta$ and $\vgamma$ (in the case of AON HMLL) and the resolution parameter $\gamma$ (in the case of GMLL) using ground truth cluster labels. 
We emphasize that this is typically not possible in practical applications, since the ground truth labels are not known. 
We make this choice in order to focus on a direct comparison of runtimes of each algorithm in a situation in which both can succeed. 
In later sections, we study the ability of HMLL and GMLL to recover known groups in synthetic and empirical data when affinities and resolution parameters are not known. 

\begin{tuftefigure}[t]
	\includegraphics[width=\textwidth]{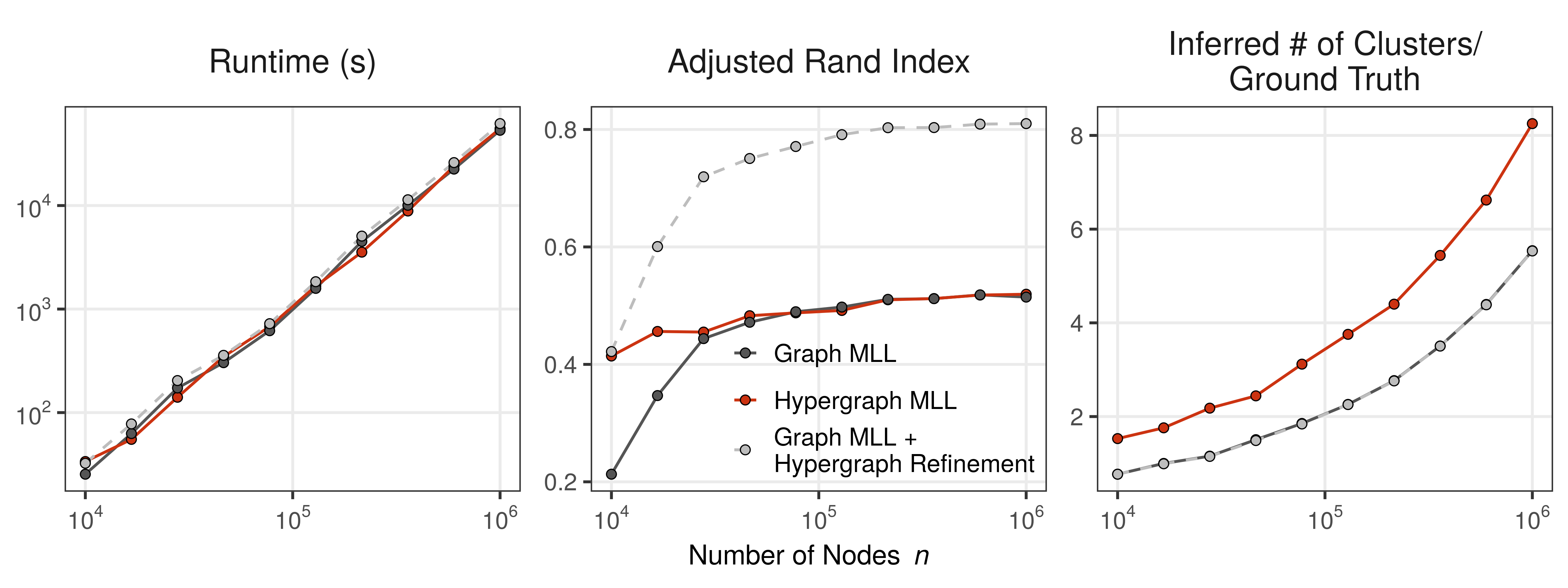}
	\caption{
		Runtime, Adjusted Rand Index, and number of clusters returned by GMLL and HMLL in a synthetic testbed with optimal affinity parameters. 
		The within-cluster edge placement probabilities are $p_2 = 3/5$, $p_3 = 1/n^3$, and $p_4 = 1/n^4$. 
		We also show in light gray the results obtained by using GMLL as a preprocessing step, whose output partition is then refined by HMLL (light gray). 
	} \label{fig:performance}
\end{tuftefigure}

\Cref{fig:performance} shows runtime, Adjusted Rand (ARI) and number of clusters returned on synthetic hypergraphs when $p_2 = 0.6$, $p_3 = 1/n^3$ and $p_4 = 1/n^4$. These parameter are chosen so that hyperedges of size three and four are rarely (if ever) contained completely inside clusters. Thus, hyperedges of different sizes provide different signal regarding ground truth clusters.
For this experiment, we implemented Graph MLL by computing a normalized clique projection, in which nodes are joined by weighted dyadic edges with weights 
\begin{align}
w_{ij} = \sum_{e: i,j \in e} \frac{1}{\abs{e} - 1}\;.
\end{align}
We also performed experiments on an unnormalized clique projection with $w_{ij} = \abs{e:i,j \in e}$, but do not show these results because, in this experiment, the associated MLL algorithm consistently fails to recover labels correlated with the planted clusters. 

On smaller instances, HMLL outperforms Graph MLL in recovering planted clusters, as measured by the ARI.
For larger instances, the recovery results are  comparable.
Interestingly, although GMLL and HMLL obtain similar accuracy in this experiment, they do so in different ways, with HMLL tending to generate more, smaller clusters than GMLL. 
Importantly, the runtimes are nearly indistinguishable, indicating that dyadic clique projections are necessary neither for accuracy nor for performance. 
We observed other choices of the parameters $p_2$, $p_3$, and $p_4$ in which HMLL substantially outperformed GMLL in cluster recovery and vice versa; however, in each case the algorithms' respective runtimes tended to differ by only a small constant factor. 

Interestingly, in this synthetic experiment, a combination of the two algorithms leads to the strongest recovery results. 
In addition to running each algorithm independently, we also ran a two-stage algorithm in which GMLL is used to generate an intermediate partition, and then HMLL is used to refine it. 
This procedure is similar to the warmstart approach from \citet{kaminski2020community}. 
We emphasize again that these results are obtained on synthetic hypergraphs with pre-optimized affinity parameters, and so the effectiveness of the refinement strategy may not generalize to real data sets. 
Indeed, in the experiments on empirical data shown in \Cref{sec:empirical}, we do not show results from the refinement procedure because the output partition was in each case essentially indistinguishable from the output of the dyadic partition. 
This may reflect the fact that we did not allow the algorithms to learn \emph{a priori} the affinity parameters associated with the true data labels. 
Further investigation into the performance of hybrid strategies would be of considerable practical importance. 

\subsection{Dyadic Projections and the Detectability Threshold}

Informally, an algorithm is able to \emph{detect} communities in a random graph model with fixed labels $\vz$ when the output labeling $\hat{\vz}$ of that algorithm  is, with probability bounded above zero, correlated with $\vz$. 
Using arguments from statistical physics, \citet{decelle2011asymptotic} conjectured the existence of a regime in the graph stochastic blockmodel in which no algorithm can successfully detect communities. 
This conjecture has since been refined and proven in various special cases; see \citet{abbe2017community} for a survey. 
In the dyadic stochastic blockmodel with two equal-sized planted communities, a necessary condition for detectability in the large-graph limit is 
\begin{align}
\frac{(c_i - c_o)^2}{2(c_i + c_o)} \geq 1\;, \label{eq:detectability_dyadic}
\end{align}
where $c_i$ is the mean number of within-cluster edges attached to a node, and $c_o$ is the mean number of between-cluster edges attached to a node. 
If this condition is not satisfied, no algorithm can reliably detect communities in the associated graph stochastic blockmodel, even though the communities are statistically distinct. 
This bound limits direct inferential methods, such as Bayesian or maximum-likelihood techniques, as well as methods based on maximization of modularity or other graph objectives \cite{nadakuditi2012graph}. 
Several recent papers have considered the detectability problem in the case of uniform hypergraphs \cite{ghoshdastidar2014consistency,Ghoshdastidar2017,Angelini2016}. 
In our model, the presence of edges of multiple sizes complicates analysis. 
Here, we limit ourselves to an experimental demonstration that the regimes of detectability for the graph SBM and our DCHSBM can differ significantly. 

\begin{tuftefigure}
	\includegraphics[width=\textwidth]{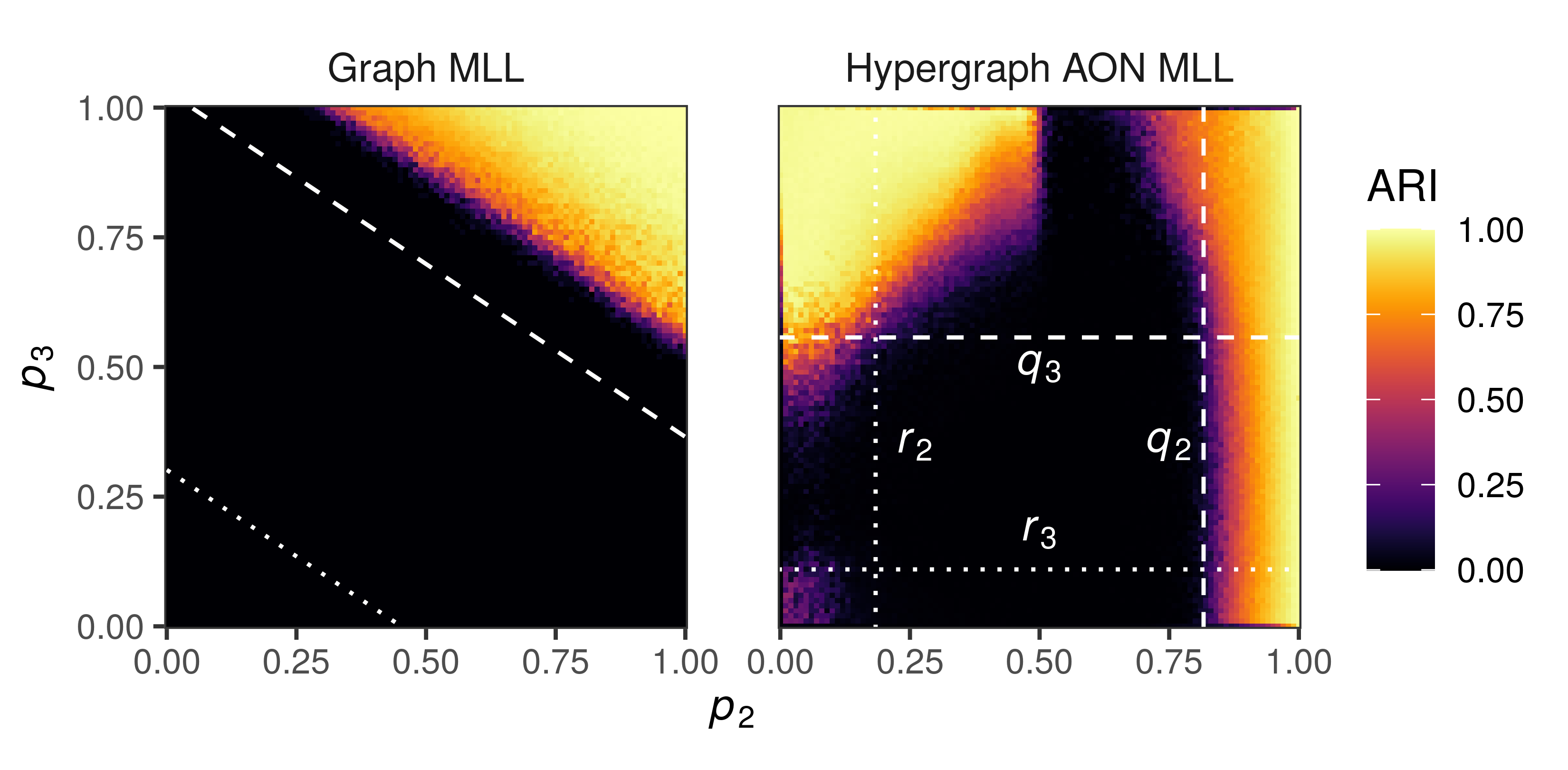}
	\caption{
		Detectability experiments in synthetic hypergraphs.
		For $i = 2,3$, $p_i$ is the proportion of within-cluster edges of size $i$.  
		Each pixel gives the mean ARI over 20 independently generated DCHSBMs of size $n = 500$ where each node is incident to, on average, five 2-edges and five 3-edges. 
		(Left) The recovered partition $\hat{\vz}$ is obtained from GMLL. 
		(Right) The recovered partition is obtained from AON HMLL (\Cref{alg:aonHMLL}). 
		The dashed  and dotted lines give various detectability thresholds as described in the main text. 
		In each panel, the returned partition $\hat{\vz}$ is highest-likelihood partition from 20 alternations between updating $\hat{\vz}$ and inference of the affinity parameters. 
		In this experiment only, we incorporate a regularization term $-n \log\bar{\ell}$ in the modularity objective in order to promote label vectors $\vz$ with fewer clusters. 
	} \label{fig:detectability}
\end{tuftefigure}

\Cref{fig:detectability} shows two experiments on a simple DCHSBM with two equal-sized communities of 250 nodes each. 
The affinity $\Omega$ is tuned so that: 
\begin{enumerate}
	\item Each node is incident to, on average, five 2-edges and five 3-edges. 
	\item A fraction $p_2$ of  2-edges join nodes in the same cluster, while a fraction $1-p_2$ of 2-edges join nodes in different clusters. 
	\item A fraction $p_3$ of 3-edges join nodes in the same cluster, while a fraction $1-p_3$ of 3-edges join nodes in different clusters. 
\end{enumerate}
In this experiment only, both GMLL and AON HMLL discussed below use the Bayesian regularization term $-n \log \bar{\ell}$ in the likelihood objective in order to encourage each algorithm to form a relatively small number of clusters. 

In the lefthand panel, we show the Adjusted Rand Index (ARI) of the returned partition $\hat{\vz}$ against the true partition $\vz$ when using the unnormalized variant of GMLL (results for the normalized variant are similar). 
This choice reflects the fact that the true number of clusters is known and is equal to 2. 
The dashed and dotted white lines give the boundaries at which \eqref{eq:detectability_dyadic} holds with equality. 
The dashed white line gives the detectability threshold for the assortative regime in which nodes are more likely to link with others in the same cluster. 
Louvain, as an agglomerative algorithm, is well-suited for detecting assortative clusterings, and is able to detect communities in much, but not all, of this regime. 
The gap between the theoretical threshold and the performance of Louvain reflects the fact that Louvain, as a stagewise greedy algorithm, possesses no optimality guarantees. 
There is also a disassortative detectable region below the dotted white line. 
The agglomerative structure of graph-based Louvain causes the algorithm to entirely fail here. 

In the righthand panel, we show the same experiment using AON HMLL. 
The dashed lines $q_2$ and $q_3$ give assortative detectability thresholds for hypothetical algorithms which entirely ignore 3-edges and 2-edges, respectively, while the dotted lines $r_2$ and $r_3$ give the corresponding disassortative thresholds.  
HMLL is able to detect the planted partition for a range of parameter values in which GMLL is not. 
These include the case in which edges of certain sizes are largely between-cluster, as shown in the top left (small $p_2$) and bottom right (small $p_3$). 
There is a regime (mid and bottom right) in which the algorithm appears to be constrained by the boundary $q_2$, suggesting that HMLL is effectively ignoring 3-edges in this regime.
As $p_3$ increases, however, 3-edges become more informative and the partition can be detected for some values $p_2 < q_2$ (top right). 
There is also a broad regime (top left) in which the hypergraph algorithm is able effectively to use both 2- and 3-edges to detect clusters, even when 2-edges are largely between-cluster.
We also observe some very limited ability of HMLL to detect clusters in the regime in which both 2-edges and 3-edges are between-cluster (bottom left). 
Since HMLL is again an agglomerative algorithm, its performance for fully disassortative partitions such as these is unreliable at best. 

Intriguingly, there are also combinations of $p_2$ and $p_3$ in which GMLL is able to detect the planted partition while HMLL is not. 
This may indicate that the pooling of edges of different sizes implied by the dyadic projection can be useful in some regimes. 
We note again that neither GMLL or HMLL are optimal inference algorithms. 
An optimal hypergraph algorithm might significantly extend the detectable regime in the right panel of \Cref{fig:detectability}. 
We pose the development of such algorithms, as well as their analysis, as highly promising avenues for future research. 

\section{Experiments with Empirical Data} \label{sec:empirical}

Next, we analyze several hypergraphs derived from empirical data.
The first two are hypergraphs of human close-proximity contact interactions~\cite{Benson-2018-simplicial},
obtained from wearable sensor data at a primary school~\cite{Stehl-2011-contact} and a high school~\cite{mastrandrea2015contact}.
Nodes are students or teachers, and a hyperedge connects groups of people that were all jointly in proximity to one another.
Node labels identify the classrooms to which each student belongs, and
the primary school data also includes a teacher associated to each class. 
Next, we created two hypergraphs from U.S.\ Congressional bill cosponsorship data~\cite{fowler2006connecting,fowler2006legislative}, where nodes correspond to congresspersons, and hyperedges correspond to the sponsor and all cosponsors of a bill in either the House of Representatives or the Senate.
We constructed another pair of data sets from the U.S.\ Congress in the form of committee memberships \cite{committees2021}. 
Each edge is a committee in a meeting of Congress, and each node again corresponds to a member of the House or a senator. 
A node is contained in an edge if the corresponding legislator was a member of the committee during the specified meeting of Congress. 
The 103rd through 115th Congresses are represented, spanning the years 1993--2017. 
There are again separate data sets for House and Senate members. 
In all congressional data sets, the node labels give the political parties of the members.
We also used a hypergraph of Walmart purchases~\cite{amburg2020clustering}, where each node is a product and a hyperedge connects a set of products that were co-purchased by a customer in a single shopping trip.
Each node has an associated product category label.
Finally, we constructed a hypergraph where nodes correspond to hotels listed at \texttt{trivago.com},
and each hyperedge corresponds to a set of hotels whose web site was clicked on by a user of Trivago within a browsing session. 
This hypergraph was derived from data released for the 2019 ACM RecSys Challenge contest.%
\footnote{\url{https://recsys.acm.org/recsys19/challenge/}}
For each hotel, the node label gives the country in which it is located.
The data sets vary in size in terms of the number of nodes, hyperedges, hyperedge sizes, and node labels (Table~\ref{tab:data sets}).

\begin{table}
	\caption{Summary of study data sets. 
		Shown are the number of nodes $n$, number of hyperedges $m$, mean degree $\bracket{d}$, standard deviation of degree $s(d)$, mean edge size $\bracket{k}$, standard deviation of edge size $s(k)$, and number of data labels $\bar{\ell}$. 
	}\label{tab:data sets}
	\centering
	\begin{tabular}{l  r  r r r r r r}
		\toprule
		&  $n$ & $m$ & $\bracket{d}$ & $s(d)$ & $\bracket{k}$ & $s(k)$ & $\bar{\ell}$ \\
		\midrule
		\texttt{contact-primary-school} & 242 & 12,704 & 127.0 & 55.3 & 2.4 & 0.6  & 11 \\
		\texttt{contact-high-school}    & 327 & 7,818 & 55.6 & 27.1 & 2.3 & 0.5  & 9 \\
		\texttt{house-bills}            & 1,494 & 43,047 &274.0 & 282.7& 9.5 & 7.2 & 2 \\
		\texttt{senate-bills}           & 293 & 20,006 &493.4 & 406.3& 7.3 & 5.5 & 2 \\
		\texttt{house-committees}       & 1,290 & 340 &9.2 & 7.1& 35.2 & 21.3 & 2 \\
		\texttt{senate-committees}      & 282 & 315 &19.0&  14.7& 17.5 & 6.6 & 2 \\
		\texttt{walmart-purchases}      & 88,860 & 65,979 &5.1 & 26.7& 6.7 & 5.3 & 11 \\
		\texttt{trivago-clicks}         & 171,495 & 220,758 & 4.0 & 7.0 & 3.2 & 2.0 & 160 \\
		\bottomrule
	\end{tabular}
\end{table}

\subsection{Model Comparison and Higher-Order Structure}

It is often stated that higher-order features are important for understanding the structure and function of complex networks. 
It is less often clarified what \emph{kinds} of higher-order features are relevant \emph{for which} networks. 
Generative modeling provides one way to compare different kinds of higher-order structure. 
In the DCHSBM, this structure is specified by the affinity function $\Omega$. 
Comparison of the likelihood functions obtained by each affinity can indicate which one is most plausible as a higher-order generative mechanism of the underlying data. 
We performed such a comparison using the symmetric affinity functions from Table~\ref{tab:symmetric_affinities}
and the labels for nodes described above. 
In this setup, we can compute an approximate ML estimate for $\Omega$, given its functional form.
In order to make concrete comparisons, it is necessary to specify the functional forms of the Group Number (GN), Relative Plurality (RP), and Pairwise (P) affinities. 
We use the following parameterizations: 
\begin{align}
\Omega(\vp) &= \omega_{{\norm{p}_0}, k} \tag{Group Number} \\ 
\Omega(\vp) &= 
\begin{cases}
\omega_{k1} & p_1 - p_2 < \frac{k}{4} \\ 
\omega_{k0} & \text{otherwise}
\end{cases} \tag{Relative Plurality} \\
\Omega(\vp) &= 
\begin{cases}
\omega_{k1} & \sum_{i \neq j} p_ip_j < \frac{k(k-1)}{4} \\ 
\omega_{k0} & \text{otherwise}
\end{cases} \tag{Pairwise} 
\end{align}
The Group Number affinity function assigns a separate parameter to each combination of edge size and number of groups. 
The Relative Plurality affinity function assigns one parameter for the case that the difference between the largest and second largest groups within an edge exceeds $k/4$, where $k$ is the size of the edge. 
The Pairwise affinity function assigns one parameter to the case that the total number of dyadic pairs in differing groups exceeds half the possible number of such pairs. 
RP, which favors edges which the two most common labels are roughly balanced in representation, is qualitatively distinct from AON, GN, and P, all of which favor edges with homogeneous cluster labels. 

Because these affinity functions possess different numbers of parameters, we compare them via the Bayesian Information Criterion (BIC) \cite{schwarz1978estimating}, which penalizes affinity functions with more parameters than are supported by the data. 
In computing the BIC, we exclude the $n$ parameters $\vtheta$, as these are the same in each model and therefore contribute an unimportant additive constant. 
The AON, RP, and P affinities each have $2\bar{k}$ parameters. 
In the case of GN, we compute the number of possible parameters for each edge size $k$ by computing the number of possible groups using the number of distinct labels in the given partition. 
For example, if the given partition contained only three distinct groups, then we do not posit parameters corresponding to edges containing more than three groups. 
It would also be reasonable to remove this restriction, in which case there would be $k$ parameters for edges of size $k$ regardless of $\vz$. 

Table~\ref{tab:comparative_omega} shows the BIC for the DCHSBM using each of these affinity functions. 
Importantly, no single affinity function is preferred across all of the study data sets, suggesting the presence of different kinds of polyadic structure. 
In the two congressional committee data sets, RP achieves the optimal BIC, while in each of the other data sets, one of the three affinities that promotes edge homogeneity is instead preferred. 
There are also important differences between these three affinities. 
In \texttt{house-bills}, the Pairwise affinity function achieves the lowest BIC overall, while in \texttt{walmart-purchases} the Pairwise affinity is preferred over all but the Group Number affinity. 
This suggests that a model involving only pairwise comparison of node labels can provide relatively strong generative explanations of the data in these cases. 
This in turn suggests that dyadic algorithms may perform at least as well on these data sets as their polyadic counterparts. 
Indeed, as we will see below, in both of these data sets, dyadic algorithms can return clusterings more correlated with ground truth than those returned by AON HMLL. 

\begin{table}
	\centering 
	\begin{tabular}{l r r r r r}
		\toprule
		& AON & GN & RP & P &  \\ 
		\midrule 
		\texttt{contact-high-school} & $2.2003$ & $\mathbf{2.1946}$ & $2.4330 $ & $2.2003 $ & $\times 10^5$ \\ 
		\texttt{contact-primary-school} & $4.1954 $ & $\mathbf{4.1646}$ & $4.3990 $ & $4.1954 $ & $\times 10^5$  \\ 
		\texttt{house-committees} & $2.7128 $  & $2.7128 $ & $\mathbf{2.7119} $ & $2.7127$ & $\times 10^5$ \\ 
		\texttt{senate-committees} & $9.7934$ & $9.7934 $ & $\mathbf{9.7736}  $ & $9.7933  $ & $\times 10^4$\\ 
		\texttt{house-bills} & $9.9719 $ &$9.9720 $ & $10.003$ & $\mathbf{9.9670}  $  & $\times 10^6$\\
		\texttt{senate-bills} & $\mathbf{3.1925} $ &$3.1926$ & $3.2030 $ & $3.1925 $ & $\times 10^6$ \\
		\texttt{walmart-purchases} & $1.0763$ &$\mathbf{1.0753}$ & $1.0806 $ & $1.0758 $ & $\times 10^6$ \\
		\texttt{trivago-clicks} & $\mathbf{1.1982} $ &$1.2015$ & $1.4364 $ & $1.2054 $ & $\times 10^7$ \\
		\bottomrule
	\end{tabular}
	\caption{
		Bayesian Information Criteria (BIC) of the DCHSBM using the the All-Or-Nothing (AON), Group Number (GN), Relative Plurality (RP), and Pairwise (P) affinity functions on our full study data sets. 
		Definitions of each affinity function are supplied in Table~\ref{tab:symmetric_affinities}.  
		Lower BIC indicates a more plausible model. 
		The affinity function achieving the lowest BIC in each data set is shown in \textbf{bold}. 
	} \label{tab:comparative_omega}
\end{table}

\subsection{Recovering Classes in Contact Hypergraphs}

To test the AON HMLL algorithm itself, we first study its behavior in the \texttt{contact-primary-school} and \texttt{contact-high-school} networks. 
The comparison of BIC scores from Table~\ref{tab:comparative_omega} suggests that GN may be the most explanatory model of the data, but we instead use AON in order to take advantage of its considerable computational benefits. 
We performed 20 alternations between AON HMLL and estimation of the AON parameters, and returned the partition with the highest DCHSBM likelihood. 
We compare the results to two dyadic methods. 
Each step of the Graph Louvain algorithm alternates between using the standard Louvain algorithm \cite{blondel2008fast} to infer clusters and estimating the resolution parameter $\gamma$ using the approximate maximum-likelihood framework of \cite{newman2016equivalence}. 
Graph Louvain returns the partition which maximizes the classical dyadic modularity objective. 
We also compare to Graph Maximum-Likelihood Louvain (GMLL), which carries out the same alternation, but instead returns the partition that maximizes the approximate log-likelihood of the corresponding planted partition stochastic blockmodel.

\begin{tuftefigure}
	\includegraphics[width=\textwidth]{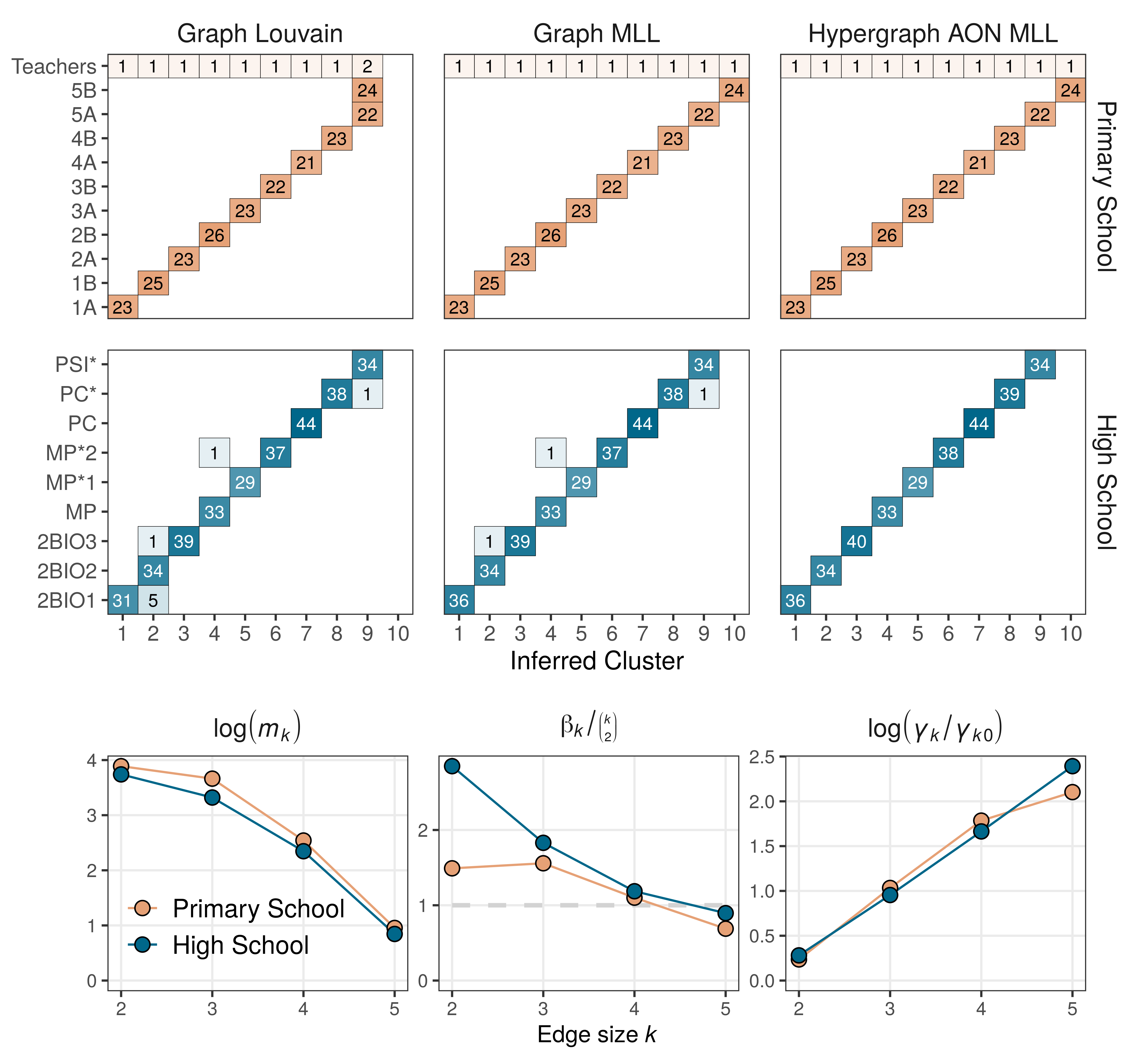}
	\caption{
		Comparison of clustering algorithms in \texttt{contact-primary-school} and  \texttt{contact-high-school}. 
		For each data set, we show a partition obtained from the classical graph Louvain modularity maximization heuristic; a partition obtained from Graph Maximum-Likelihood Louvain (GMLL); and partition obtained by AON HMLL. 
		The partition shown is the one which attains the corresponding objective function after 20 rounds of iterative likelihood maximization. 
		Each box records the number of agents with the specified combination of inferred cluster and ground truth label. 
		The bottom row visualizes the number $m_k$ of edges of size $k$, the inferred size weights $\beta_k$, and inferred resolution parameters $\gamma_k$ as defined in \eqref{eq:aon-simplified}. 
		On the far right, $\gamma_{k0} = m_k / \vol(H)^k$.  
	} \label{fig:contact-clustering}
\end{tuftefigure}

\Cref{fig:contact-clustering} compares the performance of each of these algorithms. 
In the case of \texttt{contact-primary-school}, we consider the ground truth partition to be the one that assigns exactly one teacher to each class. 
Graph Louvain is able to find partitions of students with clear correlations with the given class labels, but conflates two primary school classes and splits several high school classes (left column, top two rows). 
Graph MLL is able to perfectly recover the primary school student class labels, and misclassifies three high school students. 
Our proposed AON HMLL is able to correctly recover the given partitions in both data sets. 

We can obtain some qualitative insight into the behavior of HMLL by studying the structure of the inferred affinity function $\Omega$. 
The most intuitive way to do so is through the derived parameters $\beta_k$ and $\gamma_k$ from eq.\ \eqref{eq:aon-simplified}. 
The bottom row of \Cref{fig:contact-clustering} shows these parameters, as well as the distribution of edge sizes. 
The dependence of $\beta_k$ on edge size $k$ provides one explanation of why Graph MLL succeeds in  \texttt{contact-primary-school} but makes several errors in \texttt{contact-high-school}. 
Under the standard dyadic projection, a $k$-hyperedge generates $\binom{k}{2}$ 2-edges, and therefore appears in the dyadic modularity objective $\binom{k}{2}$ distinct times. 
In the case of \texttt{contact-primary-school}, the estimated importance parameter $\beta_k$ is indeed relatively close to $\binom{k}{2}$ (bottom center panel of \Cref{fig:contact-clustering}). 
At the optimal partition, the relative weights of edges are therefore distorted relatively little by the clique projection. 
On the other hand, the estimates for $\beta_k$ in \texttt{contact-high-school} deviate considerably from $\binom{k}{2}$, especially for $k = 2,3$.  
Here, small edges feature much more prominently in the polyadic modularity objective than they do in the projected dyadic objective, implying that the latter is a poorer approximation to the former near the optimal partition. 
This difference may explain the small errors in Graph MLL in \texttt{contact-high-school}. 
The bottom-right panel of \Cref{fig:contact-clustering} compares the inferred value of the size-specific resolution parameter $\gamma_k$ to $\gamma_{k0} = m_k/\vol(H)^k$, the implicit value used in \cite{Kami2018}. 
The inferred resolution parameters are consistently larger $\gamma_{k0}$ and increase with $k$, highlighting the value of adaptively estimating these parameters in our approach. 

\subsection{Cluster Recovery with Large Hyperedges}

In \Cref{fig:recovery_experiments}, we study the ability of AON HMLL to recover ground truth communities in several more of our study data sets.
Unlike the two contact networks, each of these data sets contains edges of size up to 25 nodes.  
We have excluded \texttt{house-committees} and \texttt{senate-committees} on the grounds that these data sets are \emph{disassortative}, indicating that AON is clearly inappropriate. 
We compare AON HMLL to two variants of GMLL. 
In the unnormalized variant, we obtain a dyadic graph by replacing each $k$-edge with a $k$-clique, thus generating a total of $\binom{k}{2}$ dyadic edges. 
In the normalized variant, we weight each edge in the $k$-clique by a factor of $\frac{1}{k-1}$. 
The normalized dyadic degree of each node is then equal to its degree in the original hypergraph. 
In either case, we then alternate between the dyadic Louvain algorithm for estimating clusters and conditional maximum likelihood inference of the resolution parameter $\gamma$. 
In each trial, we perform 20 iterations of AON HMLL as well as the two GMLL variants, returning from these the combination of group labels and parameters that achieves the highest likelihood. 
We then compare the clustering to the ground truth labels via the Adjusted Rand Index. 
We vary the maximum edge size $\bar{k}$ in order to show how each algorithm responds to the incorporation of progressively larger edges. 
Because extreme sparsity poses issues for community detection algorithms in general \cite{abbe2017community}, we show experiments for progressively denser cores of \texttt{trivago-clicks} and \texttt{walmart-purchases}. 
The $c$-core of of a hypergraph $H$ is defined as the largest sub-hypergraph $H_c$ such that all nodes in $H_c$ have degree at least $c$. 

\begin{tuftefigure}
	\includegraphics[width=\textwidth]{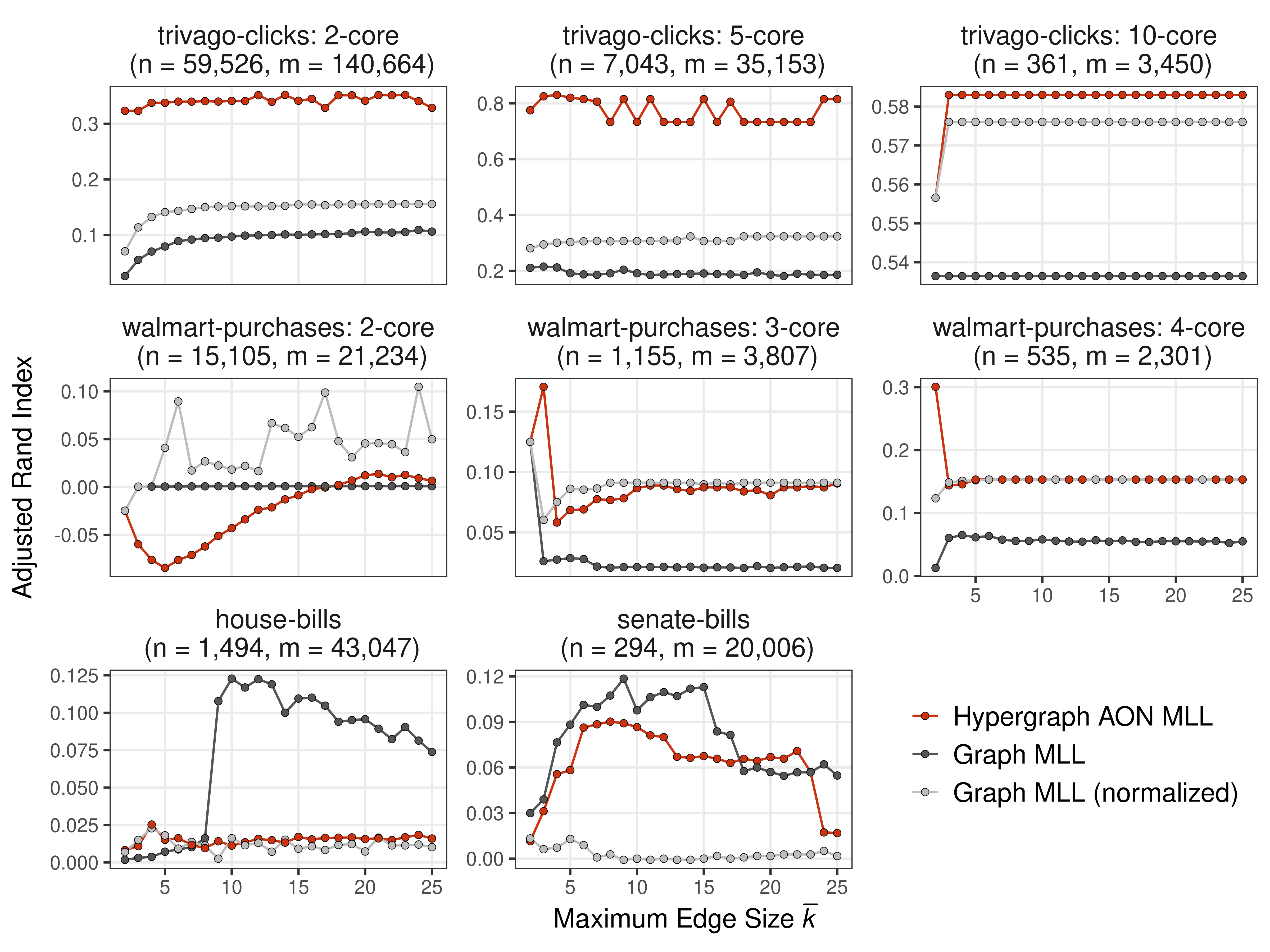}
	\caption{
		Comparison of Hypergraph All-Or-Nothing MLL \cref{alg:aonHMLL} against dyadic likelihood Louvain in data with known clusters. 
		Points give the Adjusted Rand Index of the highest-likelihood partition obtained after 20 alternations between partitioning and parameter estimation. 
		The maximum edge size $\bar{k}$ varies along the horizontal axis.
		In the panel titles, $n$ is the number of nodes and $m$ the number of edges when $\bar{k} = 25$. 
		Note that the vertical axis limits vary between panels.  
	} \label{fig:recovery_experiments}
\end{tuftefigure}

The results highlight the strong dependence of the performance of AON HMLL on the relative plausibility of the AON affinity function as a generative mechanism for the data (cf.~Table~\ref{tab:comparative_omega}). 
In \texttt{trivago-clicks}, the AON affinity function achieved the lowest BIC of all four candidates. 
Because AON is a more plausible generating mechanism by this metric, it is not surprising that AON HMLL is able to find partitions considerably more correlated with the supplied data labels than those returned by the dyadic variants. 
In \texttt{walmart-purchases}, on the other hand, the Pairwise affinity is preferred to AON. 
In this case, AON HMLL performs much worse, and in the 2-core even returns clusters which are anticorrelated with the supplied labels. 
As weakly-connected nodes are removed and the resulting data becomes denser, HMLL begins to return correlated clusters. 
However, the normalized GMLL variant is at least as effective in recovering the data labels. 
In the two Congressional bills data sets, the Pairwise affinity achieves a lower BIC than AON in the House and a comparable one in the Senate. 
Echoing this finding, a dyadic method outperforms AON HMLL in each of these cases. 
Interestingly, unnormalized GMLL performs best in \texttt{house-bills} and \texttt{senate-bills}, while normalized GMLL is preferable in \texttt{walmart-purchases}.
In addition, HMLL is the worst algorithm only in the case of the 2-core of \texttt{walmart-purchases} for small $\bar{k}$. 
HMLL may therefore be the algorithm of choice in cases when it is not known whether normalized or unnormalized dyadic representations are more appropriate for the data. 

When interpreting these recovery results, it is important to contextualize them against the limitations of community detection methods in general and of modularity maximization in particular. 
There is no ``best algorithm'' for community detection that does not make implicit assumptions about the structure of the data, and mismatch of algorithms to data sets can generate misleading results \cite{peel2017ground}.%
\footnote{That being said, a test similar to BESTest~\cite{peel2017ground}
	reveals that the likelihood under DCHSBM is much greater than the likelihood under random label permutations in most data set configurations,
	implying significant correlation between network structure and the labels.}
Even when the data generating process indeed matches algorithmic assumptions---such as a synthetic data set generating from a stochastic blockmodel---optimal algorithms may fail to detect planted communities due to sparsity \cite{decelle2011asymptotic,abbe2017community}.
Greedy modularity maximization, including the Louvain variants considered here, only finds one of possibly many local optima \cite{good2010performance}, some of which may be largely uncorrelated with each other. 
These considerations imply that (a) we cannot rule out the existence of other local optima which might achieve higher scores in any of the three algorithms and (b) the fact that an algorithm fails to recover a clustering close to the ground truth does not imply that it is ``failing'' in its stated objective, namely, local likelihood maximization. 
Overall, our results suggest that, when the assumptions of the DCHSBM with All-Or-Nothing affinity are appropriate to the data, AON HMLL can outperform dyadic approaches in recovering ground truth communities. 
In practice, since we often do not have access to ground truth labels, the question of whether or not the assumptions are appropriate to the data should be informed by domain expertise. 

\section{Discussion}

We have proposed a generative approach for clustering polyadic data, grounded in a degree-corrected hypergraph stochastic blockmodel (DCHSBM). 
From this model we have derived a symmetric, modularity-like objective, which includes the All-Or-Nothing (AON) modularity objective as an important special case. 
This derivation connects hypergraph modularity objectives to concrete modeling assumptions, which can be tuned in response to domain expertise. 
We have also formulated Louvain-like algorithms for optimizing these objectives, which are highly scalable in the case of the AON affinity function. 
Embedding this heuristic within an alternating approximate maximum likelihood scheme allows adaptive estimation of both node clusters and affinity parameters. 
We have shown experimentally that hypergraph algorithms possess markedly different detectability regimes from dyadic algorithms. 
We have also conducted experiments on empirical data, finding that hypergraph methods are preferred to dyadic ones in data sets where their modeling assumptions are well-founded. 

Our work points toward many directions of further research. 
One of these directions is algorithmic. 
Our greedy coordinate-ascent framework for inference in the DCHSBM has several important limitations. 
First, because we rely on an NP-hard optimization step, global maximization of the likelihood is never assured. 
Second, even exact maximum-likelihood itself is limited as an inference paradigm, as it uses information contained only within a small part of the likelihood landscape. 
Our method, as an approximation which is exact only when clusters are of roughly equal sizes, may also suffer from estimation bias. 
Third, the edgewise agglomerative approach embodied by Louvain-style algorithms is limited in applicability to affinity functions that promote homogeneity within edges. 
Alternative inference paradigms may ameliorate some or all of these limitations. 
Within the framework of maximum-likelihood inference, directly maximizing a profile likelihood offers an intriguing alternative to coordinate ascent \cite{bickel2009nonparametric}. 
While all maximum-likelihood methods are equivalent insofar as they optimize the same objective function, algorithmic properties such as runtime and propensity to be trapped in undesirable local optima may vary between different approaches. 
Fully Bayesian treatments \cite{peixoto2019bayesian} offer another promising path, although these are sometimes limited in their computational scalability. 
Variational belief-propagation \cite{decelle2011asymptotic, zhang2014scalable} provides an intriguing compromise, achieving considerable scalability in exchange for several approximations. 
Recent work \cite{Angelini2016} has made progress in this direction, but several questions related to scalability and behavior in nonuniform hypergraphs remain extant. 
Belief-propagation methods for scalable inference with more general affinity functions would be of particular practical interest.

There are also several important directions of theoretical development. 
One of these is the question of detectability in the DCHSBM. 
Because the DCHSBM is more flexible than the dyadic DCSBM, the theory of detectability in this model may be substantially more complex. 
Another direction concerns the properties of the dyadic modularity objective that extend to the hypergraph modularity objectives discussed here. 
In addition to its role as a comparison against null models \cite{newman2006modularity} and as a term in the DCSBM likelihood \cite{newman2016equivalence}, the dyadic modularity also expresses the stability of diffusion processes on graphs \cite{delvenne2010stability} and the energy of discrete surface tensions defined on graphs \cite{boyd2019stochastic}. 
Extensions of these properties, or explanations of why they fail to generalize, would be helpful for both theorists and practitioners. 

\subsection*{Acknowledgments}

We are grateful to Tiago de Paula Peixoto for pointing out the distinction between exact and approximate maximum likelihood estimation of $\vtheta$ and $\Omega$ as discussed at the end of \Cref{sec:symmetric_omega}. 

\subsection*{Funding}
This research was supported in part by 
ARO Award W911NF19-1-0057,
ARO MURI,
NSF Award DMS-1830274,
and JP Morgan Chase \& Co.

\subsection*{Competing Interests}
The authors declare that they have no competing interests.

\subsection*{Author Contributions}
PSC, NV, and ARB all conceived the research, developed algorithms, performed experiments, and wrote the manuscript.

\subsection*{Software and Data}

Software and data sufficient to reproduce and extend the experiments and analysis in this paper are available at the following repository:
{\center\url{https://github.com/PhilChodrow/HypergraphModularity}}. \\ 

\noindent The data are also hosted in packaged format at 
{\center\url{https://www.cs.cornell.edu/~arb/data/#hyperlabels}}.

\pagebreak
\begin{fullwidth}
    \bibliographystyle{dgleich-bib}
    \bibliography{refs}
\end{fullwidth}

\pagebreak

\begin{flushright}%
    \textbf{\MakeTextUppercase{\allcapsspacing{Supplementary Information}}} \\ 
    \textit{Generative Hypergraph Clustering: From Blockmodels to Modularity}
    \bigskip\par%
\end{flushright}\noindent%

\appendix

\section{Unidentifiability in $\theta$ and $\Omega$}
\label{sec:unidentifiable}
Fix $\vtheta$ and $\Omega$, and fix $\ell \in [\bar{\ell}]$ and $s \in \R$. If we do not provide any additional normalization condition, we can construct a modified vector $\vtheta'$ given by 
\begin{align}
\theta_i' = \begin{cases}
s\theta_i & \quad z_i = \ell \\ 
\theta_i & \quad \text{otherwise.}
\end{cases}
\end{align}
Define the modified affinity function $\Omega'(\vz_R) = {s}^{-k_{\ell}(\vz_R)}\Omega(\vz_R)$, where $k_{\ell}(\vz_R)$ is the number of times that label $\ell$ appears in $\vz_R$. 
By construction, $\cL(\vz, \Omega, \vtheta) = \cL(\vz, \Omega', \vtheta')$. 
Thus, we are unable to distinguish a single choice of $\vtheta$ and $\Omega$ from the infinite family of possibilities obtained by varying $s$. 
The normalization \eqref{eq:normalization} is used to select a single choice of parameters from this family. 

\section{Maximum-likelihood Estimation of $\theta$} 
\label{sec:proof_constant}
Recall that we choose to constrain $\vtheta$ such that 
\begin{align}
\sum_{i = 1}^n\theta_i\delta(z_i,\ell) = \vol(\ell),\;\; \ell = 1, \ldots, \bar{\ell}\;, \label{eq:normalization2}
\end{align}
where $\vol(\ell) = \sum_{i = 1}^n d_i\delta(z_i, \ell)$, with $d_i$ representing the (weighted) number of hyperedges in which node $i$ appears. Letting $k_i(R)$ be the number of times nodes $i$ show up in tuple $R$, we can write $d_i = \sum_{R \in \cR}k_i(R)a_R$.

In order to prove that the maximum likelihood estimate of $\vtheta$ is $\hat{\vtheta} = \vd$ (the formula in \cref{eq:hat-d}), we will prove the following identity
\begin{align}
\sum_{R\in \cR} b_R\pi(\vtheta_R) \Omega(\vz_R) = \sum_{k = 1}^{\bar{k}}\sum_{\vy\in[\bar{\ell}]^k}\Omega(\vy) \prod_{y \in \vy} \vol(y)\;. \label{eq:constant}
\end{align}
This implies that the second term of $Q$ does not depend on $\vtheta$ once the normalization~\eqref{eq:normalization} has been enforced. Since the first term of $Q$ is also independent of $\vtheta$, we can compute the estimate $\hat{\vtheta}$ by maximizing $K(\vtheta)$ alone. We write 
\begin{align}
K(\vtheta) = \sum_{R\in \cR}a_R \log \pi(\vtheta_R)= \sum_{R\in \cR}a_R\sum_{j \in R} \log \theta_j \;.  \label{eq:ktheta}
\end{align}
The first-order optimality conditions on $\vtheta$ from the objective~\eqref{eq:ktheta} with constraint~\eqref{eq:normalization2} 
imply that $K(\vtheta)$ is maximized with respect to $\theta_i$ when we set $\hat{\theta}_i = \sum_{R\in \cR}k_i(R) a_R = d_i$. Thus, the maximum likelihood estimate of $\vtheta$ is $\hat{\vtheta} = \vd$ as long as identity~\eqref{eq:constant} holds. 

We prove identity~\eqref{eq:constant} by direct calculation, after first introducing some useful notation. Recall that $\mathcal{R}$ continues to represent sets of \emph{unordered} node tuples from a set of $n$ nodes. We will use $\mathcal{T}$ to represent the set of \emph{ordered} node tuples. For a hyperedge size $k$, let $\mathcal{T}_k$ specifically refer to the set of ordered sets of $k$ nodes, where again we have allowed repeated nodes. For every $T \in \mathcal{T}$, $\vz_T$ is the set of cluster labels for nodes in $T$, and $\vtheta_T$ is a vector of degree parameters. Recall that $b_R$ represents the number of ways we can order elements from a set $R \in \mathcal{R}$. This is given formally by $\binom{k}{k_1(R),\ldots,k_n(R)}$, where $k_i(R)$ is the number of times that node $i$ appears in $R$. Thus, every $R \in \mathcal{R}$ is associated with $b_R$ different tuples in $\mathcal{T}$. 
We then have
\begin{align}
\sum_{R\in \cR} b_R \pi(\vtheta_R) \Omega(\vz_R) &= \sum_{k = 1}^{\overline{k}} \sum_{T\in \cT_k} \pi(\vtheta_T) \Omega(\vz_T) \\ 
&= \sum_{k = 1}^{\overline{k}} \sum_{T\in \cT_k} \left(\prod_{i = 1}^n\theta_i^{t_i}\right) \Omega(\vz_T) \\
&= \sum_{k = 1}^{\overline{k}} \sum_{\vy \in [\overline{\ell}]^k} \Omega(\vy)\sum_{T\in \cT_k} \delta(\vz_T, \vy)  \left(\prod_{i=1}^n\theta_i^{t_i}\right)\;. \\ 
&= \sum_{k = 1}^{\overline{k}} \sum_{\vy \in [\overline{\ell}]^k} \Omega(\vy)\prod_{y \in \vy} \sum_{i = 1}^n \theta_i \delta(z_i, y) \\ 
&= \sum_{k = 1}^{\overline{k}}\sum_{\vy\in[\overline{\ell}]^k}\Omega(\vy) \prod_{y \in \vy}^k \vol(y)\;,
\end{align}
as was to be shown. 

\section{Maximum-Likelihood Estimation of $\Omega$} \label{sec:omega_estimate}
We now provide details for the maximum-likelihood estimate of $\omega$ in \cref{eq:hat_omega_generalized}.
Inserting \eqref{eq:constant} into the definition of $Q$ yields 
\begin{align}
Q(\vz, \Omega, \vtheta) = \sum_{R \in \cR}a_R \log \Omega(\vz_R) - \sum_{k = 1}^{\bar{k}} \sum_{y\in [\bar{\ell}]^k} \Omega(\vy) \prod_{y\in \vy} \vol(y)\;.
\end{align}
Focusing on the set $Y$ on which $\Omega$ takes constant value $\omega$, we write 
\begin{align}
Q(\vz, \Omega, \vtheta) = \sum_{\vy \in Y}\left[\sum_{R \in \cR}\delta(\vz_R, \vy)a_R \log \omega - \omega \prod_{y\in \vy} \vol(y)\right] + C\;,
\end{align}
where $C$ captures terms that do not depend on $\omega$. 
The first-order condition reads
\begin{align}
0 = \frac{\partial Q}{\partial \omega} = \sum_{\vy \in Y}\left[ \frac{1}{\omega}\sum_{R \in \cR} a_R \delta(\vz_R, \vy) - \prod_{y \in \vy} \vol(y)\right]\;.
\end{align}
Solving for the constant $\omega$ yields \eqref{eq:hat_omega_generalized}.

\section{Derivation of \eqref{eq:combinatorial}}\label{sec:mod_cut_vol_deriv}
We have 
\begin{align}
Q(\vz, \Omega, \vtheta) &= \sum_{R\in \cR}\left[a_R\log \Omega(\vz_R) -b_R\pi(\vtheta_R)\Omega(\vz_R)  \right] \\ 
&= \sum_{R\in \cR}\sum_{\vp \in \cP}\delta(\vp, \phi(\vz_R))\left[a_R\log \Omega(\vz_R) -b_R\pi(\vtheta_R)\Omega(\vz_R)  \right] \\ 
&= \sum_{k = 1}^{\bar{k}}\sum_{\vp \in \cP_k} \left[\sum_{R \in \cR_k}a_R\delta(\vp, \phi(\vz_R))\log\Omega(\vp) - \sum_{\vy \in [\bar{\ell}]^k}\delta(\vp, \phi(\vy))\Omega(\vp)\prod_{y\in \vy} \vol(y)\right] \\ 
&= \sum_{k = 1}^{\bar{k}}\sum_{\vp \in \cP_k} \left[\cut_\vp(\vz)\log\Omega(\vp) - \vol_\vp(\vz)\Omega(\vp)\right] \\ 
&= \sum_{\vp \in \cP} \left[\cut_\vp(\vz)\log\Omega(\vp) - \vol_\vp(\vz)\Omega(\vp)\right]\;, 
\end{align}
where to obtain the third line we have applied the identity \eqref{eq:constant}.

\section{Fast evaluation of volume terms}
\label{sec:eval_sums}
Define
\begin{align}
U_\vp \equiv \mathrm{\vol}_\vp(\vz)\binom{k}{\vp}^{-1}\prod_{j = 1}^k \abs{\{h:p_h = j\}}! \;.
\end{align}
The number $U_\vp$ is an order-corrected version of the moments $\vol_\vp(\vz)$ in which there is a single term for each distinct partition vector $\vp$. 
Let $\ve_j$ be the $j$th standard basis vector. 
For each $k$, let $\mu_k = \sum_{\ell = 1}^{\bar{\ell}} \mathit{\vol}(\ell)^k$; these moments can be computed in $O(nk)$ time.
\begin{prop} \label{prop:recursion}
	Fix $\vp$, and let $r = \norm{\vp}_0$ be the number of nonzero elements of $\vp$. 
	Then,  
	\begin{align}
	U_\vp = \mu_{p_r} U_{\vp - p_r \ve_r} - \sum_{j = 1}^{r - 1} U_{\vp + p_r(\ve_j - \ve_r)}\;. 
	\end{align}
\end{prop}
%
\begin{proof}
	For notational compactness, let $v_\ell = \vol(\ell)$. 
	Let $r = \norm{\vp}_0$ be the number of nonzero elements in $\vp$. 
	By definition, we have 
	\begin{align}
	U_\vp = \binom{k}{\vp}^{-1} \prod_{j = 1}^k \abs{\{h \; : \; p_h = j\}}! \sum_{\vy \in [\bar{\ell}]^k}\delta(\vp, \phi(\vy))\prod_{y\in \vy}^k v_{y} \;.
	\end{align}
	For a given label vector $\vy$ such that $\phi(\vy) = \vp$, there are $\binom{k}{\vp} \prod_{j = 1}^k \left(\abs{\{h \; : \; p_h = j\}}!\right)^{-1}$ ways to permute the labels of $\vy$. 
	The number $U_\vp$ thus has exactly one term of the form $\prod_{y\in \vy}^k v_{y}$ for each equivalence class of cluster label vectors under permutations of indices. 
	We thus obtain the simplification 
	\begin{align}
	U_\vp = \sum_{t_1 \neq \cdots \neq t_r} \prod_{j = 1}^r v_{t_j}^{p_j}\;.
	\end{align}
	Peeling off the $r$th term in this product gives 
	\begin{align}
	U_\vp &= \sum_{t_1 \neq \cdots \neq t_{r-1}} \prod_{j = 1}^{r-1} v_{t_j}^{p_j} \sum_{\ell \notin \{t_1,\ldots,t_{r-1}\}}v_{\ell}^{p_r} \nonumber \\ 
	&= \sum_{t_1 \neq \cdots \neq t_{r-1}} \prod_{j = 1}^{r-1} v_{t_j}^{p_j} \left[\sum_{\ell = 1}^{\bar{\ell}}v_{\ell}^{p_r} - \sum_{j' =1}^{r-1}v_{t_j'}^{p_r}\right]\;. \label{eq:ugly}
	\end{align}
	We first recognize 
	\begin{align}
	\sum_{t_1 \neq \cdots \neq t_{r-1}} \prod_{j = 1}^{r-1} v_{t_j}^{p_j} = U_{\vp - p_r \ve_r} \quad \text{and} \quad \sum_{\ell = 1}^{\bar{\ell}} v_\ell^{p_r} = \mu_{p_r}\;.
	\end{align}
	We next observe that 
	\begin{align}
	\sum_{t_1 \neq \cdots \neq t_{r-1}} \prod_{j = 1}^{r-1} v_{t_j}^{p_j} \sum_{j' =1}^{r-1}v_{t_j'}^{p_j'} = \sum_{j = 1}^{r-1} U_{\vp + p_r(\ve_j - \ve_r)}\;.
	\end{align}
	Inserting these identities into \eqref{eq:ugly} completes the proof. 
\end{proof}

We note that \Cref{prop:recursion} also gives an efficient recursive update for the numbers $U_\vp$ when moving nodes between clusters. 
Let $\vz$ and $\vz'$ be distinct label vectors, and let $\{\mu_k\}$ and $\{\mu_k'\}$ be the moments of volumes evaluated with respect to $\vz$ and $\vz'$ respectively. 
Then, applying \Cref{prop:recursion} twice yields the update $\Delta U_\vp  = U'_\vp - U_\vp$: 
\begin{align}
\Delta U_{\vp} = (\Delta \mu_{p_r}) U_{\vp - p_r\ve_r} + \mu_{p_r}(\Delta U_{\vp - p_r\ve_r}) + (\Delta \mu_{p_r})(\Delta U_{\vp - p_r\ve_r}) - \sum_{j = 1}^{r-1} \Delta U_{\vp+p_r(\ve_j - \ve_r)}\;. 
\end{align}
We can thus use \Cref{prop:recursion} to compute the set of numbers $\{U_\vp\}$ once ``from scratch,'' and then efficiently update these numbers as the label vector changes in the course of optimization. 

\section{Derivation of \eqref{eq:aon-simplified}}\label{sec:mod_cut_vol_deriv_aon}
The all-or-nothing affinity function for a size-$k$ hyperedge $R$ can be written as
\begin{align}
\label{ppomega}
\Omega(\vz_R) = \omega_{k0} + \delta(\vz_R)(\omega_{k1}- \omega_{k0}) = \omega_{k0} + \delta(\vz_R) \gamma_k \beta_k,
\end{align}
and we also have
\begin{align}
\label{logppomega}
\log \Omega(\vz_R) = \log \omega_{k0} + \delta(\vz_R)(\log \omega_{k1}- \log \omega_{k0}) = \log \omega_{k0} + \delta(\vz_R) \beta_k.
\end{align}
In order to later simplify terms in the modularity objective, first recall that $\mathcal{R}_k$ is the set of unordered $k$-tuples of nodes, and $\mathcal{T}_k$ is the set of ordered $k$-tuples. Let $T \subseteq \ell$ indicate that all nodes from an ordered $k$-tuple $T \in \cT_k $ are in cluster $\ell$. 
Recall that $\pi(\vv)$ is the product of entries of a vector $\vv$.
Then we have
\begin{align}
\sum_{R \in \cR_k} b_R \pi(\vd_R) \delta(\vz_R) = \sum_{T \in \cT_k} \pi(\vd_T) \delta(\vz_T) = \sum_{\ell = 1}^{\bar{\ell}} \sum_{T \subseteq \ell} \pi(\vd_T) 
= \sum_{\ell = 1}^{\bar{\ell}} \vol(\ell)^k.
\end{align}
Using the fact that $\vtheta_R = \vd_R$, and applying definition of the cut function $\cut_k(\vz)$, we derive the simplified form of the all-or-nothing modularity objective:
\begin{align}
&Q(\vz, \Omega, \vtheta) = \sum_{R\in \cR}\left[a_R\log \Omega(\vz_R) -b_R\pi(\vtheta_R)\Omega(\vz_R)  \right] \\ 
&= \sum_{k = 1}^{\bar{k}} \sum_{R \in \cR_k} \left[a_R(\log \omega_{k0} + \delta(\vz_R)\beta_k ) - b_R \pi(\vd_R) ( \omega_{k0} + \delta(\vz_R)\gamma_k \beta_k)\right] \\
&=  \sum_{k = 1}^{\bar{k}} \sum_{R \in \cR_k} \left[a_R\log \omega_{k0}  - b_R \pi(\vd_R) \omega_{k0}\right] + \sum_{k = 1}^{\bar{k}} \beta_k \left [\sum_{R \in \cR_k} \left(a_R \delta(\vz_R)  - \gamma_k  b_R \pi(\vd_R) \delta(\vz_R)\right)\right] \\
&=  \sum_{k = 1}^{\bar{k}} \sum_{R \in \cR_k} \left[a_R\log \omega_{k0}  - b_R \pi(\vd_R) \omega_{k0}\right] + \sum_{k = 1}^{\bar{k}} \beta_k \left [m_k - \cut_k(\vz)  - \gamma_k  \sum_{\ell=1}^{\bar{\ell}} \vol(\ell)^k\right] \\
&=J(\vomega) - \sum_{k = 1}^{\bar{k}} \beta_k \left [\cut_k(\vz)  + \gamma_k  \sum_{\ell=1}^{\bar{\ell}} \vol(\ell)^k\right]
\end{align}
where
\begin{align}
J(\vomega) = \left[\sum_{k = 1}^{\bar{k}} \beta_k m_k \right] +\left[ \sum_{k = 1}^{\bar{k}} \sum_{R \in \cR_k} (a_R\log \omega_{k0}  - b_R \pi(\vd_R) \omega_{k0})\right].
\end{align}

 \section{All-Or-Nothing Hypergraph Maximum-Likelihood Louvain} \label{sec:AON_louvain}

In \Cref{sec:sym_HMLL}, we gave pseudocode for Hypergraph Maximum-Likelihood Louvain (HMLL) for general symmetric affinity functions $\Omega$ (\Cref{alg:symHMLLstep,alg:symHMLL}). 
For the special case of the All-Or-Nothing affinity function, it is possible to obtain a considerably faster Louvain algorithm. 
In particular, it is possible to follow the ``supernode'' strategy of classical dyadic Louvain, enabling to compute on data structures much simpler than the original adjacency hyperarray. 
\Cref{alg:aonHMLL} describes the outer loop of AON HMLL. 
The overall structure is much the same as in \Cref{alg:symHMLL}. 
In each iteration of the outer loop, we collapse an initial clustering $\vz$ into a reduced hypergraph $\bar{H}$ with the function $\text{Collapse}(H, \vz)$, in a way that still preserves the special AON structure. We run a Louvain step on the reduced hypergraph, and then use a function $\text{Expand}(H, \vz, \bar{\vz})$ to translate the clustering $\bar{\vz}$ on the reduced hypergraph into a clustering $\vz'$ on the original hypergraph. If $\vz$ and $\vz'$ coincide, the Louvain step found no improvement and we terminate. Otherwise, we begin a new iteration, which begins by collapsing the updated clustering to a reduced hypergraph before running the Louvain step.
In constructing the reduced representation $\bar{H}$, it is necessary to store only the collapsed hypergraph alongside a vector $\bar{\vs}$ of original sizes of each edge. 
The degree $\bar{d}_\ell$ of each collapsed node $\ell$ is simply the sum of degrees of nodes in the corresponding cluster $\ell$. 

The Louvain step itself is given by \Cref{alg:aonHMLLstep}. 
At each stage, we examine a collapsed node $i$ and compute the change in objective associated with changing $\bar{\vz}$ to $\bar{\vz}^{i \mapsto A}$, where 
\begin{align}
\bar{\vz}_j^{i \mapsto A} = 
\begin{cases}
A &\quad i = j \\ 
\bar{\vz}_j &\quad \text{otherwise}\;,
\end{cases}
\end{align}
is the collapsed label vector obtained by setting $z_i \mapsto A$ while leaving other entries unchanged. 
This calculation is contained within the subroutine $\Delta Q_{\mathrm{AON}}$ (\Cref{alg:aonmod}). 
Notably, evaluation of $\Delta Q_{\mathrm{AON}}$ is much faster than the evaluation of $\Delta Q$ in the general symmetric case \Cref{alg:symHMLLstep}.
Local changes in the volume term of~\eqref{eq:aon-simplified} requires only updating sums of powers of cluster volumes, scaled by the resolution parameters $\vbeta$ and $\vgamma$. 
Updates to the cut term of~\eqref{eq:aon-simplified} requires summing over the changes in cut status of each hyperedge adjacent to the node $i$ to be updated. 
In the AON setting, there are only two possible statuses---cut and uncut. 
Checking whether an edge is cut is much easier than computing an exact partition vector for a hyperedge before and after a potential move. 
This step is therefore much faster in practice than computing cut changes for the symmetric objective. 

\begin{algorithm}[t]
	\DontPrintSemicolon
	\SetAlgoLined
	\SetKwRepeat{Do}{do}{while}
	\KwData{Hypergraph $H$, parameter vectors $\vbeta$ and $\vgamma$}
	\KwResult{Updated label vector $\vz$}\;
	$\vz' \gets \vz \gets [n]$ \tcp*{assign each cluster to singleton} 
	\Do{$\vz \neq \vz'$}{
		$\vz \gets \vz'$\\
		$\bar{H}, \bar{\vs} \gets \text{Collapse}(H,\vz)$ \\
		$\bar{\vz}' \gets \text{AONLouvainStep}(\bar{H},\bar{\vs},\vbeta, \vgamma)$\\
		$\vz' \gets \text{Expand}(H, \vz, \bar{\vz}')$
	}
	\Return{$\vz$}
	\caption{AllOrNothingHMLL($H$, $\vbeta$, $\vgamma$)} \label{alg:aonHMLL}
\end{algorithm}

\begin{algorithm}[t]
	\SetAlgoLined
	\DontPrintSemicolon
	\KwData{Collapsed hypergraph $\bar{H} = (\bar{V}, \bar{E})$, edge size vector $\bar{\vs}$, parameter vectors $\vbeta$ and $\vgamma$}
	\KwResult{Label vector $\bar{\vz}$ on $\bar{V}$}\;
	$\bar{\vz} \gets [\bar{n}]$ \tcp*{$\bar{n}$ is the number of nodes in $\bar{H}$}

	$improving \gets \text{true}$\;
	\While{\text{improving}}{
		$improving \gets \text{false}$\;
		\For{$i\in \bar{V}$}
		{
			$\mathcal{E}_i \gets \{e \in \bar{E} \;|\; i \in e\}$ \tcp*{hyperedges incident on node $i$} 
			$\mathcal{A}_i \gets \{\ell \;| \; \exists j \in \mathcal{E}_i : z_j = \ell \}$ \tcp*{clusters adjacent to node $i$}
			
			\tcp{maximum $\Delta$ and maximizer $A$ of the change in $Q$ from moving node $i$ to cluster $A'\in \mathcal{A}_i$}
			
			$(\Delta, A) \gets \argmax_{A'\in \mathcal{A}_i}\Delta Q_{\mathrm{AON}}(\bar{H},\bar{\vz}, \bar{\vs}, A', i,\mathcal{E}_i, \vbeta, \vgamma)$\;
			\If{$\Delta>0$}{
				$\bar{z}_i \gets A$, $improving \gets \text{true}$
			}
		}	
	}
	\caption{AONLouvainStep($\bar{H}$, $\bar{\vs}$, $\vbeta$, $\vgamma_\ell$)} \label{alg:aonHMLLstep}
\end{algorithm}

\begin{algorithm}[t]
	\SetAlgoLined
	\DontPrintSemicolon
	\KwData{Collapsed hypergraph $\bar{H} = (\bar{V}, \bar{E})$, current clustering $\bar{\vz}$, edge size vector $\bar{\vs}$, proposed new cluster $A$, node to move $i$, edge set incident to $i$ $\mathcal{E}_i$, parameter vectors $\vbeta$ and $\vgamma$}
	\KwResult{Change in modularity associated with moving node $i$ to new cluster $A$}\;
	$v_A \gets \vol(A)$, $v_i \gets \vol(\bar{\vz}[i])$\\
	$\Delta v = \sum_{k = 1}^{\bar{k}}\beta_k\gamma_k \left[v_i^k - (v_i-\bar{d}_i)^k + v_A^k- (v_A+\bar{d}_i)^k\right]$\; 
	
	$\Delta c = \sum_{e \in \mathcal{E}_i} \beta_{\bar{s}_i} \left[\delta\left(\bar{\vz}_e^{i \mapsto A}\right) - \delta\left(\bar{\vz}_e\right)\right]$\;
	\textbf{return} $\Delta c +\Delta v$
	\caption{$\Delta Q_{\mathrm{AON}}(\bar{H},\;\bar{\vz}, \;\bar{\vs},\; A,\; i,\;\mathcal{E}_i, \;\vbeta, \;\vgamma)$} \label{alg:aonmod}
\end{algorithm}

\end{document}